\documentclass[12pt, draftclsnofoot, journal, letter, onecolumn]{IEEEtran}

\usepackage{graphicx}
\usepackage{epsfig}
\usepackage{latexsym}
\usepackage{amsfonts}
\usepackage{here}
\usepackage{rawfonts}
\usepackage[latin1]{inputenc}
\usepackage[T1]{fontenc}
\usepackage{calc}
\usepackage{capitalgreekitalic}
\usepackage{url}
\usepackage{enumerate}
\usepackage{color}
\usepackage[tbtags]{amsmath}
\usepackage{amssymb}
\usepackage{upref}
\usepackage{epic,eepic}
\usepackage{times}
\usepackage{dsfont}
\usepackage{comment}
\usepackage{cite}

\newcommand\independent{\protect\mathpalette{\protect\independenT}{\perp}}
\def\independenT#1#2{\mathrel{\rlap{$#1#2$}\mkern2mu{#1#2}}}

\newtheorem{theorem}{{\bf Theorem}}

\newtheorem{lemma}{{\bf Lemma}}

\newcommand{\qed}{\nobreak \ifvmode \relax \else
  \ifdim\lastskip<1.5em \hskip-\lastskip
  \hskip1.5em plus0em minus0.5em \fi \nobreak
  \vrule height0.75em width0.5em depth0.25em\fi}

\newcounter{step}
\newlength{\totlinewidth}
  {\end{list}%
  \rule{\linewidth}{1pt}}
\newcounter{substep}

  {\end{list}}

\newlength{\aligntop}
\newlength{\alignbot}
 

\IEEEoverridecommandlockouts

\begin{document}

\title{Distributed Compute-and-Forward Based Relaying Strategies in Multi-User Multi-Relay Networks}
\author{\authorblockN{Seyed Mohammad Azimi-Abarghouyi, Mohsen Hejazi, and Masoumeh Nasiri-Kenari, \emph{Senior Member, IEEE}}\\
   \thanks{
  The authors are with Electrical Engineering Department, Sharif University of Technology, Tehran, Iran. 
Emails:
\protect\url{sm_azimi@ee.sharif.edu, mhejazi@ee.sharif.edu, mnasiri@sharif.edu}. } }

\maketitle

\begin{abstract}

In this paper, we propose different practical distributed schemes to solve the rank failure problem in the compute and forward (CMF)-based multi-user multi-relay networks without central coordinator, in which the relays have no prior information about each other. First, a new relaying strategy based on CMF, named incremental compute-and-forward (ICMF), is proposed that performs quite well in terms of the outage probability. We show that the distributed ICMF scheme can even outperform the achievable rate of centralized optimal CMF in strong enough inter relay links, with much less complexity. Then, as the second scheme, amplify-forward and compute (AFC) is introduced in which the equations are recovered in the destination rather than in the relays. Finally, ICMF and AFC schemes are combined to present hybrid compute-amplify and forward (HCAF) relaying scheme, which takes advantages of both ICMF, and AFC and improves the performance of the ICMF considerably. We evaluate the performance of the proposed strategies in terms of the outage probability and compare the results with those of the conventional CMF strategy, the Decode and Forward (DF) strategy, and also the centralized optimal CMF. The results indicate the substantial superiority of the proposed schemes compared with the conventional schemes, specially for high number of users and relays.


\end{abstract}


\section{Introduction} 
In a multiuser multi-relay network, the users desire to transfer their messages to a common destination or to different destinations with the help of relays in an efficient and reliable way. To date, most proposed relaying schemes such as amplify-and-forward (AF) and decode-and-forward (DF) perform quite well in the absence of multiuser interference [1-2], where the users transmit in orthogonal channels (for instance, by using Time Division Multiple Access (TDMA)) at the cost of low network throughput. On the other hand, if the users transmit simultaneously, the performance will be degraded due to the multiuser interference or noise amplification. By utilizing network coding along with DF or AF relaying scheme, a combination of the users' messages can be constructed in each relay to improve the system throughput [3-4].

The novel relaying technique, known as compute-and-forward (CMF) [5], has been designed for multiuser applications with the aim of increasing the network throughput. In this scheme, based on a noisy received combination of simultaneously transmitted signals of the users, each relay attempts to recover a linear integer-combination of the users' messages (an equation), instead of recovering each individual message separately. To enable the relay to recover the equation, the CMF scheme is usually implemented based on using proper lattice codes [6]. An attractive characteristic of CMF scheme is that the channel state information (CSI) is not needed in the transmitters, which makes it practical for most applications. The recovered equations by the relays are then forwarded to the common destination that attempts to solve and recover the users' messages. In fact, the CMF method exploits rather than combats the interference, towards a better network performance. 

CMF method has been considered and studied in multi antenna systems [7], two way relaying systems [8-9], cooperative distributed antenna systems [10], multi-access relay channels [11], generalized multi-way relay channels [12], and two transmitter multi-relay systems [13]. However, among the most important challenges of the CMF is the rank failure problem which has not been solved yet. Since each relay selects its equation coefficients independently to maximize its own rate, the equations that are received from different relays at the destinations can be linearly dependent. In other words, the coefficient matrix of the equations received by the destinations can encounter rank failure. In this case, the destination cannot recover the users' messages and the system performance deteriorates considerably [13]. In the method proposed in [14], each relay recovers its best $T_{\text{max}}$ equations with the highest rates and then sends them to the destination, each equation in a time slot. These equations are not necessarily independent. Among the equations received from the relays, the destination selects $L$ ($L$ is the number of transmitted messages) independent equations, if any, with highest rate in order to recover the messages. However, this method can also encounter rank failure, although its probability reduces with the increase of ${T_{\text{max}}}$ at the cost of very high required time slots from relays to destination. The rank failure problem is mostly considered when there exists a central coordinator with a global CSI that computes independent equations with a maximum computation rate (the equation detecting rate), then allocates the equations to the relays [10]. On the other hand, most AF- and DF-based strategies, which do not encounter the rank failure problem, have only a simple timer implemented in each relay to coordinate different relays [15-16].

In this paper, we propose novel distributed strategies to handle the rank failure problem for a general multi-user multi-relay wireless network at the absence of a central coordinator. As demonstrated, the proposed schemes are based on simple timer on each relay and can compete with the conventional AF and DF schemes practically. Our contributions are as follows;

1. We propose a relaying strategy based on CMF, named "incremental compute-and-forward" (ICMF), in which the linearly independent equations are recovered one by one through cooperation among the relays using a simple timer in each relay. For the first time, we propose a general distributed successive method for recovering different number of equations. We provide an algorithm (Algorithm 1) that can be implemented in multi-relay scenarios with low complexity. We present a receiver structure for our scheme, and propose limited search area for its integer optimization problem (Lemma 1). 

Here, though we use the same concept as in [17] for creating effective channels in our scheme, the successive CMF presented in [17] is a one-stage successive equation computator implemented in one relay. In fact, this scheme is proposed for recovering an equation with the help of another decoded equation at a rate higher than the CMF.

2. We prove that despite of its much less complexity, the ICMF with sufficiently strong inter-relays channels outperforms the achievable rate of the optimal centralized CMF scheme with global knowledge (Theorem 2). 

3. We introduce "amplify-forward-and-compute" (AFC), based on using the conventional amplify-and-forward relaying method and the integer-forcing linear receiver (IFLR) introduced by Zhan, et al [7]. In AFC, each relay simply amplifies its received combination of the users' noisy signals and forwards the result to the destination. Then, the destination recovers all the required equations. Hence, the relay structure in this scheme is considerably simpler than those in CMF and ICMF. In AFC method, the destination acts as a computation center, while in the CMF and the ICMF, the computation (recovering the equations) is performed by the relays in a distributed manner. Here, we borrow the reciever structure for IFLR from [7], with slight modification to be matched with the amplified received signals. We also introduce limited search area for its integer optimization problem.

4. We introduce "hybrid compute-amplify and forward" (HCAF) scheme, based on the combination of ICMF and AFC schemes. In this strategy, first the linearly independent equations are successively recovered by the relays based on the ICMF scheme till the stage at which the maximum computation rate derived by the relay is less than the target rate. Then for the rest of the equations, the AFC technique is used and the related equations are recovered by the destination. The destination exploits the equations transmitted from the computing relays to recover the rest of equations with higher rate. A new receiver structure and limited search area for the integer optimization problem of the HCAF scheme are presented as well.

We evaluate the performance of the proposed schemes, in terms of outage probability, and compare the results with the CMF and the DF relaying schemes. Our numerical results show that the ICMF and AFC strategies outperform the CMF method significantly and provide higher diversity orders. The proposed ICMF slightly outperforms the centralized optimal CMF with global knowledge at strong inter-relays channels, while holds an acceptable performance degradation compared with the optimal CMF at very poor inter-relays channels. Although the AFC strategy performs worse than the ICMF, especially when the links between the relays are strong, it has much less complexity. HCAF improves the performance of the ICMF at the cost of more complex receiver structure at the destination. Our proposed schemes lead to substantially less outage probability compared to the conventional DF scheme. Moreover, the performance gain of the proposed schemes increases with the number of users and relays.

The reminder of the paper is organized as follows. In Section II, the system model and the conventional CMF strategy are described. The proposed methods, namely ICMF, AFC, and HCAF, are presented in Section III, Section IV, and Section V, respectively. Numerical results are given in Section VI. Finally, Section VII concludes the paper.

\textbf{Notations:} The operators ${(\mathbf{A})^*}$, $(\mathbf{A})^T$, $||\mathbf{A}||$, and $\text{span}(\mathbf{A})$ stand for conjugate transpose, transpose, frobenius norm, and the space constructed from the column vectors of matrix $\mathbf{A}$, respectively. The symbol $\left| x \right|$ is the absolute value of the scalar $x$, while ${\rm{lo}}{{\rm{g}}^ + }\left( x \right)$ denotes ${\rm{max}}\left\{ {\log \left( x \right),0} \right\}$. $\mathbb{E}\{\cdot\}$ is the expectation operator and $\independent$ indicates the linear independency of vectors. $\mathbf{I}$ denotes identity matrix.

\section{System Model and Related Work} 
\subsection{System Model} We consider a multi-user multi-relay cooperative network, shown in Fig. 1 [18], consisting of  $L$ users, $M$ relays and one common destination. There are no direct links between the users and the destination. Each user $i$ exploits a lattice encoder, with power constraint $P_T$, to map its message $w_i$ to a complex-valued codeword $x_i$ of length $n$ with ${\left| {\left| {{x_i}} \right|} \right|^2} \le n{P_T}$. We denote the received signals at relay $m$ by $y_m^r$ and at the destination from the relay $m$ by $y_m$. The power constraint at each relay is $P_R$. The element $h_{im}$ of the channel matrix $\mathbf{H}$ represents the channel coefficient from user $i$ to relay $m$, $f_m$ indicates the channel coefficient from relay $m$ to the destination, and $g_{ab}$ denotes the channel coefficient from relay $a$ to relay $b$. The channel coefficient $h_{im}$, $f_{m}$, and $g_{ab}$ are assumed to be independent complex Gaussian distributed random variables with the variances $\sigma _{im}^2$, $\sigma _m^2$, and $\sigma _{r,ab}^2$, respectively. Moreover, block fading is assumed, where the channels are considered to be constant during the transmission periods required for message exchanges. We assume that each relay has only information about its own channels and is not aware of the other relays' channel states. 

At the first time slot, all $L$ users transmit their codewords simultaneously to the relays. In the following $L$ or $M$ time slots, depending on the schemes used, $L$ or $M$ signals are transmitted by the relays, each in separate dedicated slot. The received signals $y_m^r$ and $y_m$ at the relay $m$ and at the destination can respectively be written as,
\begin{eqnarray}
y_m^r = \mathop \sum \limits_{i = 1}^L {h_{im}}{x_i} + {z_m}
\end{eqnarray} 		
\begin{eqnarray}
{y_m} = {f_m}x_m^r + {\eta _m}
\end{eqnarray}													
where $z_m$ and ${\eta _m}$ are independent additive white Gausian noises with identical variances equal to $N_0$, and $x_m^r$ denotes the signal transmitted by the relay $m$.
\begin{figure}[tb!]
\centering

\includegraphics[width =5in]{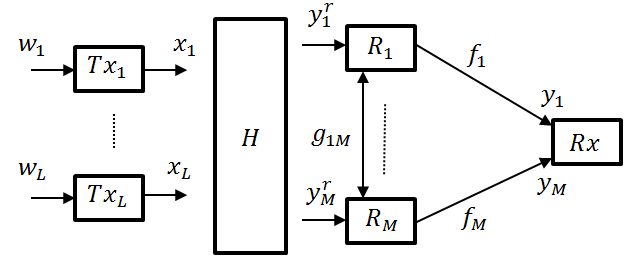}

\caption{Multi-User Multi-Relay Cooperative Network}

\end{figure}

\subsection{Conventional Compute-and-Forward Strategy} In the conventional compute-and-forward (CMF) method [5,13], in the first time slot, all the $L$ users transmit their own codewords $x_{i}$, $i = 1, \ldots ,{\rm{ }}L$, simultaneously to the relays. Based on its received signal $y_m^r$, each relay $m$, $m=1,...,M$, attempts to detect an equation $s_m$, a linear combination of users' codewords, with complex integer equation coefficients vector (ECV) ${\mathbf{a}_m} = {\left[ {{a_{1m}}, \ldots ,{a_{Lm}}} \right]^*} \in {\left( {\mathbb{Z} + i\mathbb{Z}} \right)}^L$, i.e., ${s_m} = \mathop \sum \limits_{i = 1}^L {a_{im}}{x_i} = \mathbf{a}_m^*\mathbf{X}$, where the vector $\mathbf{X} = {\left[ {{x_1}, \ldots ,{x_L}} \right]^*}$ includes the codewords of all users.
The coefficient vector in each relay is selected based on maximizing the relay's computation rate, i.e., the rate of detecting the equation $s_m$, as follows [5]:
\begin{eqnarray}
{\mathbf{a}_m} = arg{\rm{mi}}{{\rm{n}}_{{\mathbf{a}_l} \in {{\left( {\mathbb{Z}+ i\mathbb{Z}} \right)}^L},{\mathbf{a}_l} \ne 0}}\left( {\mathbf{a}_l^*{\mathbf{H}_m}{\mathbf{a}_l}} \right)
\end{eqnarray}
where $\text{SNR}_T = {P_T}/{N_0}$. The vector $\mathbf{h}_m$ and the matrix $\mathbf{H}_m$ are defined as,
\begin{eqnarray}
{\mathbf{h}_m} \buildrel \Delta \over = {\left[ {{h_{1m}}, \ldots ,{h_{Lm}}} \right]^*}
\end{eqnarray}
\begin{eqnarray}
{\mathbf{H}_m} \buildrel \Delta \over = \mathbf{I} - \frac{{\text{SNR}_T}}{{1 + \text{SNR}_T{{\left| {\left| {{\mathbf{h}_m}} \right|} \right|}^2}}}{\mathbf{h}_m}\mathbf{h}_m^*
\end{eqnarray}
An efficient algorithm for solving the above integer optimization problem has been proposed in [15].  To detect the equation $s_m$, the relay $m$ quantizes the scaled received signal ${\alpha _m}y_m^r$ to its nearest lattice point ${s_m} = Q\left( {{\alpha _m}y_m^r} \right)$, where [5]
\begin{eqnarray}
{\alpha _m} = \frac{{\text{SNR}_T\mathbf{h}_m^*{\mathbf{a}_m}}}{{1 + \text{SNR}_T{{\left| {\left| {{\mathbf{h}_m}} \right|} \right|}^2}}}
\end{eqnarray}									
and $Q(.)$ denotes the lattice quantizer function. The achievable computation rate of $s_m$ is equal to [5]:	
\begin{eqnarray}
{r_m} = {\log ^ + }({\left( {\mathbf{a}_m^*{\mathbf{H}_m}{\mathbf{a}_m}} \right)^{ - 1}}).
\end{eqnarray}
The $M$ equations, $s_m, m=1,...,M$, each independently detected by one of the relays, are then orthogonally transmitted with power $P_r$ to the destination in the next $M$ consecutive time slots.
Since the channel from the relay $m$ to the destination is a simple point-to-point channel, according to (2), the transmission rate over this channel is,
\begin{eqnarray}
{\tilde r_m} = {\rm{log}}\left( {1 + \text{SNR}_R{{\left| {{f_m}} \right|}^2}} \right)
\end{eqnarray}											
where $\text{SNR}_R = {P_R}/{N_0}$. This rate is achievable by using the CMF strategy in one user case [5].
Hence, the overall rate for recovering equation $s_m$ at the destination is determined by,
\begin{eqnarray}
{R_m} = \min \left( {{r_m},{{\tilde r}_m}} \right).
\end{eqnarray}
									
The destination receives $M$ equations from the relays. To recover the users' messages, the destination should select $L$ equations from these $M$ equations. This can be done in $\left( {\begin{array}{*{20}{c}}M\\L\end{array}} \right)$ different ways. Let $S_u$ denotes the $u$-th set of selected equations, i.e.
\begin{eqnarray}
{S_u} = \left\{ {{u_1}, \ldots ,{u_L}} \right\};u \in \left\{ {1,2, \ldots ,\left( {\begin{array}{*{20}{c}}M\\L\end{array}} \right)} \right\}
\end{eqnarray}
where $u_i$ indicates the $i$-th equation in the set $S_u$, and the matrix $\mathbf{A}_u$ denotes the ECVs corresponding to the equations in $S_u$. For each set $S_u$, the symmetric achievable rate ${R_{{s_u}}}$, i.e. the rate of recovering all the messages, is 
\begin{eqnarray}
R_{s_u} =\left\{ {\begin{array}{*{20}{c}}{\min( R_{u_1}, \ldots , R_{u_L} ) \hspace{20pt}, \hspace{5pt} rank( \mathbf{A}_u ) = L}\\ 0 \hspace{70pt} , \hspace{20pt} O.W. \end{array}} \right. 
\end{eqnarray}

It is noteworthy that if the ranks of all possible sets of equations $S_u$ are less than $L$, a rank failure is occurred and the destination cannot recover the messages; which leads to an outage event. Therefore, the achieved rate of the CMF method can be written as
\begin{eqnarray}
R_{\text{CMF}} = \frac{L}{M + 1}\rm{max}\left( R_{s_1}, \ldots ,R_{s_{\small{\left( {\begin{array}{*{20}{c}}M\\L\end{array}}\right)}}} \right),
\end{eqnarray}
where the coefficient $\frac{L}{{M + 1}}$ is due to the fact that in CMF method the transmission and the recovering of the users' messages at the destination take place in $M+1$ time slots.

\section{Incremental Compute-and-Forward (ICMF)}
In this method, $L$ independent equations with the highest computation rates are recovered and sent to the destination through the cooperation among the relays in a distributed manner. First, each relay calculates its own overall computation rate (using (9)). Then, the relay with the highest computation rate transmits its recovered equation to the destination, which is received by the other relays as well. For the second equation, each relay again computes another equation independent from the first one, with the maximum computation rate. Among them, the relay with the highest second computation rate transmits its derived equation, which is again received by the other relays as well. This process is repeated until all $L$ equations are derived and transmitted to the destination. That is, to recover a new equation, each relay finds an equation with the maximum rate, which is linearly independent of the previously computed and transmitted equations, and then the relay with the highest rate at that stage is selected to transmit the new recovered equation, as described in the following.  In each relay, the previously received equations are exploited in each stage to increase the rate of recovering the new equation, the concept first introduced in [17].

Specifically, the ICMF can be described as follows. At stage $k$, the $k$-th best equation is recovered and transmitted to the destination in the corresponding time slot, as follows. Each relay knows the $k-1$ best equations, $s_{\text{max}}^1$,$\ldots$,$s_{\text{max}}^{k-1}$ , that are transmitted in the previous $k-1$ time slots simply by listening and detecting the signals transmitted in the earlier slots. In our performance evaluation, we consider the possible failure at each relay in detecting the $k-1$ previously transmitted equations. Let's define the matrices
\begin{eqnarray}
{\mathbf{E}_k} = \left[ {\begin{array}{*{20}{c}}{\mathbf{e}_1^*}\\ \vdots \\{\mathbf{e}_{k - 1}^*}\end{array}} \right],{\mathbf{S}_k} = \left[ {\begin{array}{*{20}{c}}{{s_{\text{max}}^1}}\\ \vdots \\{{s_{\text{max}}^{k-1}}}\end{array}} \right]
\end{eqnarray}
where $\mathbf{e}_i$,$i=1$,$\ldots$,$k-1$, is the ECV of the $i$-th transmitted equation. Therefore, we can write,
\begin{eqnarray}
{\mathbf{S}_k} = {\mathbf{E}_k}\mathbf{X}
\end{eqnarray}
where $\mathbf{X} = {\left[ {{x_1}, \ldots ,{x_L}} \right]^*}$. Assume that the equation $s_{\text{max}}^j, j=1,...,k-1$, is computed and transmitted by relay $n_j$. The rate of receiving this equation at relay $m$ is,
\begin{eqnarray}
r_m^{{e_j}} = {\rm{min}}\left( {r_{{n_j}}^j,{r_{m{n_j}}}} \right)
\end{eqnarray}
where ${r_{m{n_j}}}$ is the rate of the point-to-point channel between relays $m$ and $n_j$ as follows:
\begin{eqnarray}
{r_{m{n_j}}} = {\rm{log}}\left( {1 + \text{SNR}_R{{\left| {{g_{m{n_j}}}} \right|}^2}} \right)
\end{eqnarray}
$r_{{n_j}}^j$ is the computation rate of recovering this equation in relay $n_j$, will be given in (27). Hence, the overall achievable rate of all of the $k-1$ previously transmitted equations at relay $m$ is,
\begin{eqnarray}
r_m^{e,k} = {\rm{mi}}{{\rm{n}}_{j = 1, \ldots ,k - 1}}\left( {r_m^{{e_j}}} \right),
\end{eqnarray}
Now, e.g. the $m$-th, relay attempts to recover a new equation based on the $k-1$ equations correctly received in the previous time slots and its own received signal $y_m^r$.
First, the effect of previous equations is removed from the received signal $y_m^r$ using projection space method [19] as:
\begin{eqnarray}
\hat y_m^k = y_m^r - \mathbf{h}_m^*\mathbf{E}_k^*{\left( {{\mathbf{E}_k}\mathbf{E}_k^*} \right)^{ - 1}}{\mathbf{S}_k}
\end{eqnarray}
This makes the optimum ECV derivation simpler, which will be shown later in Theorem 1. Then, a controlled and desired linear combination of this signal and the previously derived equations is made as follows: 
\begin{eqnarray}
\tilde y_m^k = \beta _m^k\hat y_m^k + \mathbf{c}{{_m^k}^*}{\mathbf{S}_k}.
\end{eqnarray}
The coefficients of the linear combination (19) are selected in a way to maximize the computation rate of the relay (the rate of recovering an equation from signal $\tilde y_m^k$) that is independent from the $k-1$ previously transmitted equations. 

From (18-19), to recover the equation $\mathbf{a}_l^*\mathbf{X}$ from $\tilde y_m^k$, we rewrite (19) as,
\begin{eqnarray}
\tilde y_m^k = \mathbf{a}_l^*\mathbf{X} + \left( {\beta _m^k\mathbf{g}{{_m^k}^*} + \mathbf{c}{{_m^k}^*}{\mathbf{E}_k} - \mathbf{a}_l^*} \right)\mathbf{X} + \beta _m^k{z_m}
\end{eqnarray}
where $\mathbf{g}_m^k$  is defined as:
\begin{eqnarray}
\mathbf{g}{{_m^k}^*} \buildrel \Delta \over = {\mathbf{h}_m}^*\left( {\mathbf{I} - \mathbf{E}_k^*{{\left( {{\mathbf{E}_k}\mathbf{E}_k^*} \right)}^{ - 1}}{\mathbf{E}_k}} \right)
\end{eqnarray}
The effective noise variance for this equation is,
\begin{eqnarray}
{N_{\text{eq}}} =\mathbb{E}\left\{ {{{\left| {\tilde y_m^k - \mathbf{a}_l^*\mathbf{X}} \right|}^2}} \right\}= {\left| {\beta _m^k} \right|^2} + \text{SNR}_T{\left| {\left| {\beta _m^k\mathbf{g}_m^k + \mathbf{E}_k^*\mathbf{c}_m^k - {\mathbf{a}_l}} \right|} \right|^2}
\end{eqnarray}
where $\text{SNR}_T = {P_T}/{N_0}$. Hence, the computation rate of this equation according to the above effective noise variance is given by\footnote{Note that an equation with message transmission power of $P$ and effective recovery noise variance of $N_{\text{eq}}$ has computation rate equal to ${\rm{lo}}{{\rm{g}}^ + }\left( {\frac{P}{{{N_{\text{eq}}}}}} \right)$ [5].}
\begin{eqnarray}
{\rm{r}}_m^k = {\rm{lo}}{{\rm{g}}^ + }\left( {\frac{{\text{SNR}_T}}{{{{\left| {\beta _m^k} \right|}^2} + \text{SNR}_T{{\left| {\left| {\beta _m^k\mathbf{g}_m^k + \mathbf{E}_k^*\mathbf{c}_m^k - {\mathbf{a}_l}} \right|} \right|}^2}}}} \right)
\end{eqnarray}

From (23), to obtain the maximum computation rate, we should solve the following optimization problem:
\begin{eqnarray}
\mathop {\max }\limits_{\beta _m^k,\mathbf{c}_m^k} {\rm{lo}}{{\rm{g}}^ + }\left( {\frac{{\text{SNR}_T}}{{{{\left| {\beta _m^k} \right|}^2} + \text{SNR}_T{{\left| {\left| {\beta _m^k\mathbf{g}_m^k + \mathbf{E}_k^*\mathbf{c}_m^k - {\mathbf{a}_l}} \right|} \right|}^2}}}} \right)
\end{eqnarray}
The following theorem presents the solution of this optimization problem:
\begin{theorem}
In stage $k$, the optimum values of $\beta _m^k$ and vector $\mathbf{c}_m^k$ for recovering the equation with coefficient vector $\mathbf{a}_l$ at relay $m$ are, respectively,
\begin{eqnarray}
\beta _{m,\text{opt}}^k = \frac{\mathbf{g}{{_m^k}^*} {\mathbf{a}_l}} {\frac{1}{{\text{SNR}_T}} + {|| \mathbf{g}_m^k ||^2}}
\end{eqnarray}
\begin{eqnarray}
\mathbf{c}{{_{m,\text{opt}}^k}^*} = {\mathbf{a}_l^*}{\mathbf{E}_k^*}{\left( {\mathbf{E}_k}{\mathbf{E}_k^*} \right)^{ - 1}}
\end{eqnarray}
\end{theorem}
\begin{proof}
The proof is given in Appendix I.
\end{proof}

By substituting (25) and (26) in (23), the computation rate of the equation with the coefficient vector $\mathbf{a}_l$, at relay $m$ and at stage $k$, is computed as
\begin{eqnarray}
r_m^k = \text{min}\{ \log^ + ({(\mathbf{a}_l^*\mathbf{V}_m^k{\mathbf{a}_l})^{- 1}}),r_m^{e,k}\} 
\end{eqnarray}
where $r_m^{e,k}$ is given in (17) and
\begin{eqnarray}
\mathbf{V}_m^k \buildrel \Delta \over =\mathbf{I} - \frac{{\mathbf{g}_m^k\mathbf{g}{{_m^k}^*}}}{{\frac{1}{{\text{SNR}_T}} + {{\left| {\left| {\mathbf{g}_m^k} \right|} \right|}^2}}} - \mathbf{E}_k^*{\left( {{\mathbf{E}_k}\mathbf{E}_k^*} \right)^{ - 1}}{\mathbf{E}_k}
\end{eqnarray}
The relay $m$, at stage $k$, has to find the equation with the highest possible rate in (27) that is linearly independent from the previous $k-1$ equations, i.e., $\mathbf{e}_1$,$\ldots $,$\mathbf{e}_{k-1}$. Hence, from (27), the relay $m$ finds its optimum ECV based on the following optimization problem
\begin{eqnarray}
\mathbf{a}_m^k = {\rm{mi}}{{\rm{n}}_{{\mathbf{a}_l} \in {{\left( {\mathbb{Z} + i\mathbb{Z}} \right)}^L}}}\left( {\mathbf{a}_l^*\mathbf{V}_m^k{\mathbf{a}_l}} \right)\nonumber
\end{eqnarray}
subject to
\begin{eqnarray}
{\rm{\text{rank}}}\left( {\left[ {{\mathbf{a}_l},{\mathbf{e}_1}, \ldots ,{\mathbf{e}_{k - 1}}} \right]} \right) = k
\end{eqnarray}
Through the following lemma, we can omit the useless ECVs, the computation rates of which are zero. The Lemma is of interest because it gives a very smaller search area for solving (29).
\begin{lemma}
To find the optimum ECV $\mathbf{a}_m^k$ in problem (29), it is sufficient to check the space of all integer vectors $\mathbf{a}_l$ with norm satisfying
\begin{eqnarray}
{\left| {\left| {{\mathbf{a}_l}} \right|} \right|^2} \le \frac{1}{{\frac{{\frac{1}{{\text{SNR}_T}}}}{{\frac{1}{{\text{SNR}_T}} + {{\left| {\left| {\mathbf{g}_m^k} \right|} \right|}^2}}} - \left| {\left| {\mathbf{E}_k^*{\left( {{\mathbf{E}_k}\mathbf{E}_k^*} \right)^{ - 1}}{\mathbf{E}_k}} \right|} \right|}}
\end{eqnarray}
\end{lemma}
\begin{proof}
From (27), in stage $k$, the computation rate of relay $m$ is zero for all $\mathbf{a}_l$ satisfying
\begin{eqnarray}
{\rm{\mathbf{a}}}_l^*\mathbf{V}_m^k{\mathbf{a}_l} \ge 1
\end{eqnarray}
From (28), we can rewrite the left side of (31) as
\begin{eqnarray}
{\rm{\mathbf{a}}}_l^*\mathbf{V}_m^k{\mathbf{a}_l}= {\left| {\left| {{\mathbf{a}_l}} \right|} \right|^2} - \frac{{{{\left| {\mathbf{g}{{_m^k}^*}{\mathbf{a}_l}} \right|}^2}}}{{\frac{1}{{\text{SNR}_T}} + {{\left| {\left| {\mathbf{g}_m^k} \right|} \right|}^2}}} -\mathbf{a}_l^*{\mathbf{E}_k^*{\left( {{\mathbf{E}_k}\mathbf{E}_k^*} \right)^{ - 1}}{\mathbf{E}_k}}\mathbf{a}_l
\end{eqnarray}
Using Cauchy-Schwarz inequality, ${\left| {\mathbf{g}{{_m^k}^*}{\mathbf{a}_l}} \right|}^2 \le {{{\left| {\left| {{\mathbf{a}_l}} \right|} \right|}^2}{{\left| {\left| {\mathbf{g}_m^k} \right|} \right|}^2}}$ and $\mathbf{a}_l^*{\mathbf{E}_k^*{\left( {{\mathbf{E}_k}\mathbf{E}_k^*} \right)^{ - 1}}{\mathbf{E}_k}}\mathbf{a}_l \\\le {\left| {\left| {{\mathbf{a}_l}} \right|} \right|^2}{\left| {\left| {{\mathbf{E}_k^*{\left( {{\mathbf{E}_k}\mathbf{E}_k^*} \right)^{ - 1}}{\mathbf{E}_k}}} \right|} \right|}$, we have
\begin{eqnarray}
{\rm{\mathbf{a}}}_l^*\mathbf{V}_m^k{\mathbf{a}_l} &\le & {\left| {\left| {{\mathbf{a}_l}} \right|} \right|^2} - \frac{{{{\left| {\left| {{\mathbf{a}_l}} \right|} \right|}^2}{{\left| {\left| {\mathbf{g}_m^k} \right|} \right|}^2}}}{{\frac{1}{{\text{SNR}_T}} + {{\left| {\left| {\mathbf{g}_m^k} \right|} \right|}^2}}} - {\left| {\left| {{\mathbf{a}_l}} \right|} \right|^2}{\left| {\left| {{\mathbf{E}_k^*{\left( {{\mathbf{E}_k}\mathbf{E}_k^*} \right)^{ - 1}}{\mathbf{E}_k}}} \right|} \right|}\nonumber\\&=& {\left| {\left| {{\mathbf{a}_l}} \right|} \right|^2}\left( {\frac{{\frac{1}{{\text{SNR}_T}}}}{{\frac{1}{{\text{SNR}_T}} + {{\left| {\left| {\mathbf{g}_m^k} \right|} \right|}^2}}} - {{\left| {\left| {{\mathbf{E}_k^*{\left( {{\mathbf{E}_k}\mathbf{E}_k^*} \right)^{ - 1}}{\mathbf{E}_k}}} \right|} \right|}}} \right)
\end{eqnarray}
Hence, 
\begin{eqnarray}
{\left| {\left| {{\mathbf{a}_l}} \right|} \right|^2} \ge \frac{1}{{\frac{{\frac{1}{{\text{SNR}_T}}}}{{\frac{1}{{\text{SNR}_T}} + {{\left| {\left| {\mathbf{g}_m^k} \right|} \right|}^2}}} - \left| {\left| {\mathbf{E}_k^*{\left( {{\mathbf{E}_k}\mathbf{E}_k^*} \right)^{ - 1}}{\mathbf{E}_k}} \right|} \right|}} \Rightarrow {\rm{\mathbf{a}}}_l^*\mathbf{V}_m^k{\mathbf{a}_l} \ge 1
\end{eqnarray}
\end{proof}

The equation corresponding to this ECV is recovered by quantizing $\tilde y_m^k$ in (19) as
\begin{eqnarray}
s_m^k = Q\left( {\tilde y_m^k} \right)
\end{eqnarray}
$\beta _m^k$ and $\mathbf{c}_m^k$ in (19) are substituted from (25) and (26). The overall rate of recovering this equation at the destination is
\begin{eqnarray}
{\rm{R}}_m^k = {\rm{min}}\left( {r_m^k,{{\tilde r}_m}} \right)
\end{eqnarray}
where ${\tilde r_m}$ and $r_m^k$ are given in (8) and (27), respectively. Now, the relay with the highest rate $R_m^k$ sends its equation to the destination at the $k$-th time slot, by using the technique similar to the one presented in [15], as follows. The $m$-th relay sets a timer with the value $T_m$ proportional to the inverse of its rate $R_m^k$, which counts down to zero simultaneously. The relay with timer reaching zero first has the highest rate, broadcasts a flag to inform the other relays, and then transmits its equation in the $k$-th time slot. In the stage $k$, we denote $s_{\text{max}}^k= \mathbf{e}_k^*\mathbf{X}$ as the transmitted equation by the best relay $n_k$.

After $L$ stages, $L$ independent equations are recovered and sent to the destination, and based on them, all users' messages are decoded at the destination. Assume that for the equation $L$, relay $n_L$ is selected as the best relay. It can be easily observed that the achievable rate of the proposed scheme, ICMF, (for recovering all users' messages at the destination) is,
\begin{eqnarray}
{R_{\text{ICMF}}} = \frac{L}{{L + 1}}R_{{n_L}}^L
\end{eqnarray}
Here, $R_{{n_L}}^L$, which denotes the rate of relay $n_L$, is obtained from (36). Note that $L+1$ time slots are required to transmit $L$ complex equations, in contrast to the CMF that requires $M+1$ time slots. The ICMF prosedure is summerized in Table I.

In this algorithm, when the cooperation among the relays is not possible because of the poor quality of the inter-relays channels, only the best relay with the maximum computation rate is selected using the technique described above, and that relay transmits its $L$ best independent equations to the receiver. On the other hand, for sufficiently strong inter-relays channels such that the computation rates of all previously transmitted equations at the relays are not decreased, ICMF outperforms the centralized optimal CMF with global knowledge of all links states.

\begin{theorem}
In the case of sufficiently strong inter-relays channels, ICMF achieves higher overall rate than the centralized optimal CMF.
\end{theorem}
\begin{proof}
The proof is given in Appendix II.

\end{proof}

\begin{table}[t]
  \centering
  \caption{
    Algorithm 1: ICMF Procedure}
        \vspace*{-1em}
    \begin{tabular}{l}
      \hline
   For stage $k=1,...,L$
\\
\hspace*{+10pt} I. For relay $m=1,...,M$
\\
\hspace*{+20pt} 1. Recovering the equation transmitted in the time slot $k-1$
\\
\hspace*{+20pt} 2. Finding ECV $\mathbf{a}_m^k$ by solving (29)
\\
\hspace*{+20pt} 3. Recovering the equation with ECV $\mathbf{a}_m^k$ by quantizing (35) 
\\
\hspace*{+18pt}   end
\\
\hspace*{+10pt} II. Best relay selection by timer setting
\\
\hspace*{+10pt} III. Best relay transmission
\\
   end
   \vspace*{.5em}
   \\  
      \hline
    \end{tabular}

\end{table}
\section{Amplify-Forward-and-Compute (AFC)}
In the first time slot, like CMF, all $L$ users transmit their own codewords simultaneously to the relays. The $m$-th relay amplifies its received signal $y_m^r$ by a gain ${\gamma _m}$ and then forwards the amplified signal to the destination in its dedicated time slot. Hence, similar to CMF, AFC requires $M+1$ time slots to transmit the $L$ messages.
According to the power constraint $P_R$ at each relay $m$, ${\gamma _m}$ is computed as,
\begin{eqnarray}
{\gamma _m} = \sqrt {\frac{{\text{SNR}_R}}{{\text{SNR}_T{{\left| {\left| {{\mathbf{h}_m}} \right|} \right|}^2} + 1}}} 
\end{eqnarray}								
where $\text{SNR}_R = {P_R}/{N_0}$ and $\text{SNR}_T = {P_T}/{N_0}$. From (1) and (2), the received signal from each relay $m$, at the destination is given by,
\begin{eqnarray}
{y_m} = {f_m}{\gamma _m}y_m^r + {\eta _m} = {f_m}{\gamma _m}\mathop \sum \limits_{i = 1}^L {h_{im}}{x_i} + {f_m}{\gamma _m}{z_m} + {\eta _m} , m = 1, \ldots ,M
\end{eqnarray}												
By defining the following vectors,
\begin{eqnarray}
\mathbf{Y} &=& {\left[ {{y_1}, \ldots ,{y_M}} \right]^*},\mathbf{X} = {\left[ {{x_1}, \ldots ,{x_L}} \right]^*},\gamma = {\left[ {{\gamma _1}, \ldots ,{\gamma _M}} \right]^*},\mathbf{f} = {\left[ {{f_1}, \ldots ,{f_M}} \right]^*},\nonumber\\\mathbf{Z} &=& {\left[ {{\eta _1}, \ldots ,{\eta _m}} \right]^*},{\mathbf{Z}^r} = {\left[ {{z_1}, \ldots ,{z_M}} \right]^*}
\end{eqnarray}
And also,
\begin{eqnarray}
\mathbf{F} \buildrel \Delta \over = diag\left( \mathbf{f} \right) \times diag\left( \gamma \right)
\end{eqnarray}
The set of signals (39) from all relays can be rewritten in the following matrix form
\begin{eqnarray}
\mathbf{Y} = \mathbf{F}\mathbf{H}\mathbf{X} + \mathbf{F}{\mathbf{Z}^r} + \mathbf{Z}
\end{eqnarray}

Because of the similarity of (42) with a point-to-point MIMO channel, we utilize the integer-forcing linear receiver (IFLR) introduced by Zhan, et al [7], with slight modifications to include the effect of amplified noise. The receiver structure is shown in Fig. 2. Similar to [7], the receiver exploits a projection matrix ${\mathbf{B}_{L \times M}}$ to recover $L$ independent equations with the complex integer coefficient matrix  ${\mathbf{A}_{L \times L}}$. By taking the same steps as in [7], the optimum projection matrix $\mathbf{B}$ can be written as
\begin{eqnarray}
{\mathbf{B}_{\text{opt}}} = \mathbf{A}{\mathbf{H}^*}{\mathbf{F}^*}{\left( {\frac{1}{{\text{SNR}_T}}\left( {\mathbf{I} + \mathbf{F}{\mathbf{F}^*}} \right) + \mathbf{F}\mathbf{H}{\mathbf{H}^*}{\mathbf{F}^*}} \right)^{ - 1}}
\end{eqnarray}
and the optimum computation rate for recovering an equation with ECV $\mathbf{a}_l$ is obtained by
\begin{eqnarray}
{R_l} = {\rm{lo}}{{\rm{g}}^ + }\left( {{{(\mathbf{a}_l^*\mathbf{V}{\mathbf{a}_l})}^{ - 1}}} \right)
\end{eqnarray}
where,
\begin{eqnarray}
\mathbf{V} \buildrel \Delta \over = \mathbf{I} - {\mathbf{H}^*}{\mathbf{F}^*}{\left( {\frac{1}{{\text{SNR}_T}}\left( {\mathbf{I} + \mathbf{F}{\mathbf{F}^*}} \right) + \mathbf{F}\mathbf{H}{\mathbf{H}^*}{\mathbf{F}^*}} \right)^{ - 1}}
\end{eqnarray}

If $L$ independent ECVs $\mathbf{a}_1,...,\mathbf{a}_L$, forming the coefficient matrix $\mathbf{A}$, are used, the AFC rate for recovering all users' messages is given by,
\begin{eqnarray}
{{\rm{R}}_{\text{AFC}}} = \frac{L}{{M + 1}}{\rm{min}}\left( {{R_1}, \ldots ,{R_L}} \right)
\end{eqnarray}
where $R_l$ is the computation rate of $\mathbf{a}_l$ given by (44). Note that due to linear independency of ECVs $\mathbf{a}_1,...,\mathbf{a}_L$, the rank failure problem is solved.
To maximize the rate $R^{\text{AFC}}$, from (46) and (44), the optimum coefficient matrix $\mathbf{A}^{\text{opt}}$ is calculated as,
\begin{eqnarray}
{\mathbf{A}^{\text{opt}}} = arg\mathop {\max }\limits_{\mathbf{A} \in {{\left( {\mathbb{Z} + i \mathbb{Z}} \right)}^{L \times L}}} \mathop {\min }\limits_{l = 1, \ldots ,L} {\rm{log}}({\left( {\mathbf{a}_l^*\mathbf{V}{\mathbf{a}_l}} \right)^{ - 1}})\nonumber\\ = arg\mathop {\min }\limits_{\mathbf{A} \in {{\left( {\mathbb{Z} + i\mathbb{Z}} \right)}^{L \times L}}} \mathop {\max }\limits_{l = 1, \ldots ,L} \left( {\mathbf{a}_l^*\mathbf{V}{\mathbf{a}_l}} \right)\nonumber
\end{eqnarray}
subject to,
\begin{eqnarray}
\left\{ {\begin{array}{*{20}{c}}{\mathbf{A} = \left[ {\begin{array}{*{20}{c}}{\mathbf{a}_1^*}\\ \vdots \\{\mathbf{a}_L^*}\end{array}} \right]}\\{\det \left( \mathbf{A} \right) \ne 0}\\{{\mathbf{a}_l} \in {{\left( {\mathbb{Z} + i \mathbb{Z}} \right)}^L},l = 1, \ldots ,L}\end{array}} \right.
\end{eqnarray}	
Following the same method as in Lemma 1, we can limit the check space to
\begin{eqnarray}
{\left| {\left| {{\mathbf{a}_l}} \right|} \right|^2} \le \frac{1}{1-\left|\left|{\mathbf{H}^*}{\mathbf{F}^*}{\left( {\frac{1}{{\text{SNR}_T}}\left( {\mathbf{I} + \mathbf{F}{\mathbf{F}^*}} \right) + \mathbf{F}\mathbf{H}{\mathbf{H}^*}{\mathbf{F}^*}} \right)^{ - 1}}\right|\right|}
\end{eqnarray}				
Finally, by solving the set of the equations, the users' messages are recovered. 
\begin{figure}[tb!]
\centering

\includegraphics[width =5in]{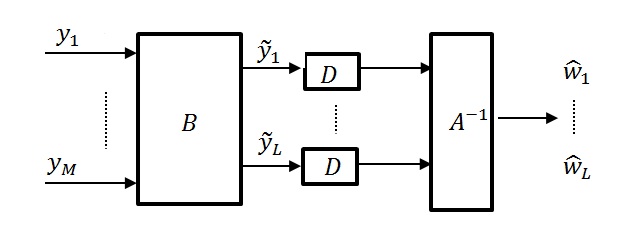}

\caption{Receiver structure at the destination for AFC method ($D$ indicates lattice decoder)}

\end{figure}
\section{Hybrid Compute-Amplify-and-Forward (HCAF)}
In ICMF, when the computation rate of the best relay in a stage $k$ is lower than the target rate $R_t$, the system encounters an outage event. In this case, recovering and sending an equation by the best relay cannot help the destination. As an alternative method, when in stage $k$, the highest computation rates of all relays are less than the target rate, the $L-k+1$ best relays, with the highest rates (though all less than $R_t$), utilize the AFC strategy to amplify and forward their received signals to the destination in the remaining $L-k+1$ time slots. These signals can then be exploited by the destination with the help of the  $k-1$ previously received equations to recover the remaining equations, which can induce higher rates compared to the ICMF. The best relays are selected by utilizing count-down timers described in Section III for the ICMF strategy.

Suppose that the relays $n_1,...,n_{k-1}$, named as computing relays, have recovered and transmitted equations ${\mathbf{d}_1^*}\mathbf{X} = {s_1^{\text{CF}}}$,$...$, $\mathbf{d}_{k - 1}^*\mathbf{X} = {s_{k - 1}^{\text{CF}}}$, or in the matrix form of 
${\mathbf{S}^{\text{CF}}} = \mathbf{D}\mathbf{X}$, where $\mathbf{D} \buildrel \Delta \over = \left[ {\begin{array}{*{20}{c}}{\mathbf{d}_1} \hdots {\mathbf{d}_{\left( {k - 1} \right)}}\end{array}} \right]^*$ and ${\mathbf{S}^{\text{CF}}} = \left[ {\begin{array}{*{20}{c}}{{{s_1^{\text{CF}}}}} \hdots {{{s_{\left( {k - 1} \right)}^{\text{CF}}}}}\end{array}} \right]^T$, in the first $k-1$ slots to the destination, and at the stage $k$, the highest rate is less than the target rate. Also, assume that at this stage, the relays $n_k,...,n_L$ are selected as amplifying relays, based on their computation rates (which are higher than the rates of the other relays). For amplifying relays, we define the following vectors 
\begin{eqnarray}
{\mathbf{Y}^{\text{AF}}} &=& {\left[ {{y_{{n_k}}}, \ldots ,{y_{{n_L}}}} \right]^*},{\bf{\gamma} ^{\text{AF}}} = {\left[ {{\gamma _{{n_k}}}, \ldots ,{\gamma _{{n_L}}}} \right]^*},{\mathbf{f}^{\text{AF}}} = {\left[ {{f_{{n_k}}}, \ldots ,{f_{{n_L}}}} \right]^*},\nonumber\\{\mathbf{Z}^{\text{AF}}} &=& {\left[ {{\eta _k}, \ldots ,{\eta _L}} \right]^*},{\mathbf{Z}^r}^{\text{AF}} = {\left[ {{z_{{n_k}}}, \ldots ,{z_{{n_L}}}} \right]^*}
\end{eqnarray}
and matrices,
\begin{eqnarray}
{\mathbf{H}^{\text{AF}}} = \left[ {\begin{array}{*{20}{c}}{\mathbf{h}_{{n_k}}^*}\\ \vdots \\{\mathbf{h}_{{n_L}}^*}\end{array}} \right]
\end{eqnarray}
\begin{eqnarray}
{\mathbf{F}^{\text{AF}}} \buildrel \Delta \over = diag\left( {{\mathbf{f}^{\text{AF}}}} \right) \times diag\left( {{\bf{\gamma} ^{\text{AF}}}} \right)
\end{eqnarray}                                                                                           
By the above definitions, the received signals from the amplifying relays at the destination can be simply written as
\begin{eqnarray}
{\mathbf{Y}^{\text{AF}}} = {\mathbf{F}^{\text{AF}}}{\mathbf{H}^{\text{AF}}}\mathbf{X}+ {\mathbf{F}^{\text{AF}}}{\mathbf{Z}^r}^{\text{AF}} + {\mathbf{Z}^{\text{AF}}}
\end{eqnarray}

The block diagram of the receiver at the destination is shown in Fig. 3. As shown in the figure, the effects of the received equations from the computing relays, i.e., relays $n_1,..., n_{k-1}$, are first removed from the signals received by the amplifying relays using projection space method,
\begin{eqnarray}
{\mathbf{\hat Y}^{\text{AF}}} = {\mathbf{Y}^{\text{AF}}} - {\mathbf{F}^{\text{AF}}}{\mathbf{H}^{\text{AF}}}{\mathbf{D}^*}{\left( {\mathbf{D}{\mathbf{D}^*}} \right)^{ - 1}}{\mathbf{S}^{\text{CF}}}
\end{eqnarray}
As mentioned previously for the ICMF, this makes the later derivations simpler (see proof of Theorem 3). The destination exploits two projection matrices ${\mathbf{B}_{\left( {L - k + 1} \right) \times \left( {L - k + 1} \right)}}$ and ${\mathbf{C}_{\left( {k - L + 1} \right) \times \left( {k - 1} \right)}}$ for the signals received from the amplifying relays and the equations received from the computing relays, respectively. After the projections, the results are added as
\begin{eqnarray}
{\mathbf{\tilde Y}^{\text{AF}}} = \mathbf{B}{\mathbf{\hat Y}^{\text{AF}}} + \mathbf{C}{\mathbf{S}^{\text{CF}}}
\end{eqnarray}
where $\mathbf{B} = \left[ {\begin{array}{*{20}{c}}{\mathbf{b}_1} \hdots {\mathbf{b}_{L - k + 1}}\end{array}} \right]^*$ and $\mathbf{C} \buildrel \Delta \over = \left[ {\begin{array}{*{20}{c}}{{\mathbf{c}_1}} \hdots {{\mathbf{c}}_{L - k + 1}}\end{array}} \right]^*$. From ${\mathbf{\tilde Y}^{\text{AF}}}$, the remaining $L-k+1$ linearly independent equations are recovered with the complex integer coefficient matrix $\mathbf{A}_{\left( {L - k + 1} \right) \times L}^{\text{AF}}$, ${\mathbf{A}^{\text{AF}}} = \left[ {\begin{array}{*{20}{c}}{\mathbf{a}_1} \hdots {\mathbf{a}_{L - k + 1}}\end{array}} \right]^*$, as follows. 

The $l$-th row of the vector ${\mathbf{\tilde Y}^{\text{AF}}}$ in (54) is given by 
\begin{eqnarray}
{\tilde y_l}^{\text{AF}} = \mathbf{b}_l^*{\mathbf{\hat Y}^{\text{AF}}} + \mathbf{c}_l^*{\mathbf{S}^{\text{CF}}} = \mathbf{b}_l^*\mathbf{G}\mathbf{X} + \mathbf{c}_l^*\mathbf{D}\mathbf{X} + \mathbf{b}_l^*{\mathbf{F}^{\text{AF}}}{\mathbf{Z}^r}^{\text{AF}} + \mathbf{b}_l^*{\mathbf{Z}^{\text{AF}}}
\end{eqnarray}
where $\mathbf{G}$ is defined as follows;
\begin{eqnarray}
\mathbf{G} \buildrel \Delta \over = {\mathbf{F}^{\text{AF}}}{\mathbf{H}^{\text{AF}}}\left( {\mathbf{I} - {\mathbf{D}^*}{{\left( {\mathbf{D}{\mathbf{D}^*}} \right)}^{ - 1}}\mathbf{D}} \right)
\end{eqnarray}
The equation $\mathbf{a}_l^*\mathbf{X}$ is recovered from ${\tilde y_l}$ as
\begin{eqnarray}
{\tilde y_l}^{\text{AF}} = \mathbf{a}_l^*\mathbf{X} + \left( {\mathbf{b}_l^*\mathbf{G} + \mathbf{c}_l^*\mathbf{D} - \mathbf{a}_l^*} \right)\mathbf{X} + \mathbf{b}_l^*{\mathbf{F}^{\text{AF}}}{\mathbf{Z}^r}^{\text{AF}} + \mathbf{b}_l^*{\mathbf{Z}^{\text{AF}}}
\end{eqnarray}
The effective noise variance in this computation is equal to
\begin{eqnarray}
{N_{\text{eq}}} ={\rm{\mathbb{E}}}\left\{ {{{\left| {{\tilde y}_l^{\text{AF}} - \mathbf{a}_l^*\mathbf{X}} \right|}^2}} \right\}= {\left| {\left| {{\mathbf{b}_l}} \right|} \right|^2} + {\left| {\left| {{\mathbf{F}^{\text{AF}}}^*{\mathbf{b}_l}} \right|} \right|^2} + \text{SNR}_T{\left| {\left| {{\mathbf{G}^*}{\mathbf{b}_l} + {\mathbf{D}^*}{\mathbf{c}_l} - {\mathbf{a}_l}} \right|} \right|^2}
\end{eqnarray}
Hence, from (57) and (58), the computation rate of the equation is given by
\begin{eqnarray}
{R_l} = {\rm{lo}}{{\rm{g}}^ + }\left( {\frac{{\text{SNR}_T}}{{{{\left| {\left| {{\mathbf{b}_l}} \right|} \right|}^2} + {{\left| {\left| {{\mathbf{F}^{\text{AF}}}^*{\mathbf{b}_l}} \right|} \right|}^2} + \text{SNR}_T{{\left| {\left| {{\mathbf{G}^*}{\mathbf{b}_l} + {\mathbf{D}^*}{\mathbf{c}_l} - {\mathbf{a}_l}} \right|} \right|}^2}}}} \right)
\end{eqnarray}
For a maximum computation rate, we should solve the following optimization problem:
\begin{eqnarray}
\mathop {\max }\limits_{{\mathbf{b}_l},{\mathbf{c}_l}} {\rm{lo}}{{\rm{g}}^ + }\left( {\frac{{\text{SNR}_T}}{{{{\left| {\left| {{\mathbf{b}_l}} \right|} \right|}^2} + {{\left| {\left| {{\mathbf{F}^{\text{AF}}}^*{\mathbf{b}_l}} \right|} \right|}^2} + \text{SNR}_T{{\left| {\left| {{\mathbf{G}^*}{\mathbf{b}_l} + {\mathbf{D}^*}{\mathbf{c}_l} - {\mathbf{a}_l}} \right|} \right|}^2}}}} \right)
\end{eqnarray}
The following theorem presents the solution of this optimization problem:
\begin{theorem}
The optimum values of vectors $\mathbf{b}_l$ and $\mathbf{c}_l$ for recovering an equation with coefficient vector $\mathbf{a}_l$ are given by
\begin{eqnarray}
\mathbf{b}_{\text{opt},l}^* = \mathbf{a}_l^*{\mathbf{G}^*}{\left( {\frac{1}{{\text{SNR}_T}}\left( {\mathbf{I} + {\mathbf{F}^{\text{AF}}}{\mathbf{F}^{\text{AF}}}^*} \right) + \mathbf{G}{\mathbf{G}^*}} \right)^{ - 1}}
\end{eqnarray}
and
\begin{eqnarray}
\mathbf{c}_{\text{opt},l}^* = \mathbf{a}_l^*{\mathbf{D}^*}{\left( {\mathbf{D}{\mathbf{D}^*}} \right)^{ - 1}}
\end{eqnarray}
Therefore, the matrices $\mathbf{B}$ and $\mathbf{C}$ can be written as
\begin{eqnarray}
\mathbf{B} = {\mathbf{A}^{\text{AF}}}{\mathbf{G}^*}{\left( {\frac{1}{{\text{SNR}_T}}\left( {\mathbf{I} + {\mathbf{F}^{\text{AF}}}{\mathbf{F}^{\text{AF}}}^*} \right) + \mathbf{G}{\mathbf{G}^*}} \right)^{ - 1}}
\end{eqnarray}
and
\begin{eqnarray}
\mathbf{C} = {\mathbf{A}^{\text{AF}}}{\mathbf{D}^*}{\left( {\mathbf{D}{\mathbf{D}^*}} \right)^{ - 1}}
\end{eqnarray}
\end{theorem}

\begin{proof}
The proof is given in Appendix III.
\end{proof}

By substituting (61) and (62) in (59), the computation rate can be written as
\begin{eqnarray}
{R_l} = {\rm{min}}\left( {{\rm{log}}\left( {{{\left( {\mathbf{a}_l^*\mathbf{U}{\mathbf{a}_l}} \right)}^{ - 1}}} \right),R_{{n_{k - 1}}}^{k - 1}} \right)
\end{eqnarray} 
where,
\begin{eqnarray}
\mathbf{U} \buildrel \Delta \over = \mathbf{I} - {\mathbf{G}^*}{\left( {\frac{1}{{\text{SNR}_T}}\left( {\mathbf{I} + {\mathbf{F}^{\text{AF}}}{\mathbf{F}^{\text{AF}}}^*} \right) + \mathbf{G}{\mathbf{G}^*}} \right)^{ - 1}}\mathbf{G} - {\mathbf{D}^*}{\left( {\mathbf{D}{\mathbf{D}^*}} \right)^{ - 1}}\mathbf{D}
\end{eqnarray}
and $R_{{n_{k - 1}}}^{k - 1}$ denotes the computation rate of relay $n_{k-1}$ as the best relay at stage $k-1$.
It is clear that the rate of the recovered remaining equations in (65) is lower than the rate of the computing relays. Hence, the rate of this strategy can be written as
\begin{eqnarray}
{R_{\text{HCAF}}} = \frac{L}{{L + 1}}{\rm{min}}\left( {{R_1}, \ldots ,{R_{\left( {L - k+1} \right)}}} \right)
\end{eqnarray}
Note that, the same as in the ICMF, the HCAF needs $L+1$ time slots to transmit $L$ messages. To maximize the computation rate  (67), from (65), $L-k+1$ linearly independent equations, which should also be independent from the computing equations, should be found from the following optimization problem,
\begin{eqnarray}
\mathbf{A}^{\text{opt,AF}} =arg\mathop {\min }\limits_{{\mathbf{A}^{\text{AF}}} \in {( {\mathbb{Z} + i \mathbb{Z}})^{\left( {L - k+1} \right) \times L}}} \mathop {\max }\limits_{l = 1, \ldots ,L-k+1} \left( {\mathbf{a}_l^*\mathbf{U}{\mathbf{a}_l}} \right)\nonumber
\end{eqnarray}
subject to
\begin{eqnarray}
\left\{ {\begin{array}{*{20}{c}}{{\mathbf{A}^{\text{AF}}} = \left[ {\begin{array}{*{20}{c}}{\mathbf{a}_1^*}\\ \vdots \\{\mathbf{a}_{L - k + 1}^*}\end{array}} \right]}\\{\det \left( {\left[ {{\mathbf{A}^{\text{AF}}};\mathbf{D}} \right]} \right) \ne 0}\\{{\mathbf{a}_l} \in {{\left( {\mathbb{Z} + i \mathbb{Z}} \right)}^L}, l = 1, \ldots ,L - k + 1}\end{array}} \right.
\end{eqnarray}
Following the same method as in Lemma 1, we can limit the check space to
\begin{eqnarray}
{\left| {\left| {{\mathbf{a}_l}} \right|} \right|^2} \le \frac{1}{1-\left|\left|{\mathbf{G}^*}{\left( {\frac{1}{{\text{SNR}_T}}\left( {\mathbf{I} + {\mathbf{F}^{\text{AF}}}{\mathbf{F}^{\text{AF}}}^*} \right) + \mathbf{G}{\mathbf{G}^*}} \right)^{ - 1}}\mathbf{G}\right|\right| - \left|\left|{\mathbf{D}^*}{\left( {\mathbf{D}{\mathbf{D}^*}} \right)^{ - 1}}\mathbf{D}\right|\right|}
\end{eqnarray}

The coefficient matrix corresponding to all equations at the destination can be written as
\begin{eqnarray}
{\mathbf{A}^{\text{opt}}} = \left[ {\begin{array}{*{20}{c}}{{\mathbf{A}^{\text{opt,AF}}}}\\\mathbf{D}\end{array}} \right]
\end{eqnarray}
The projection matrices $\mathbf{B}$ and $\mathbf{C}$ are calculated by substituting the matrices ${\mathbf{A}^{\text{opt,AF}}}$ and $\mathbf{D}$ in (63) and (64). The remaining equations, i.e. the rows of $\mathbf{S}^{\text{AF}}$, are recovered by quantizing ${\mathbf{\tilde Y}^{\text{AF}}}$ as
\begin{eqnarray}
{\mathbf{S}^{\text{AF}}} = Q\left( {{{\mathbf{\tilde Y}}^{\text{AF}}}} \right) = {\mathbf{A}^{\text{AF,opt}}}\mathbf{X}
\end{eqnarray}
where,
\begin{eqnarray}
{\mathbf{\tilde Y}^{\text{AF}}} = \mathbf{B}{\mathbf{\hat Y}^{\text{AF}}} + \mathbf{C}{\mathbf{S}^{\text{CF}}}
\end{eqnarray}
Finally, by solving the $L$ independent equations, obtained from the $k-1$ computing relays and the $L-k+1$ amplifying relays, the destination can recover all of the users' messages.

\begin{figure}[tb!]
\centering

\includegraphics[width =6in]{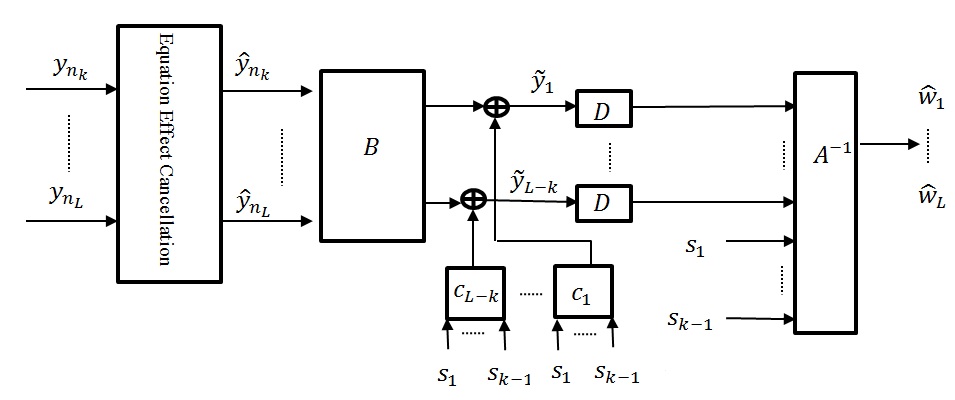}
\caption{Receiver structure at the destination for HCAF method ($D$ shows lattice decoder)}

\end{figure}

\begin{figure}[t]
\centering
\center
\includegraphics[width =6in]{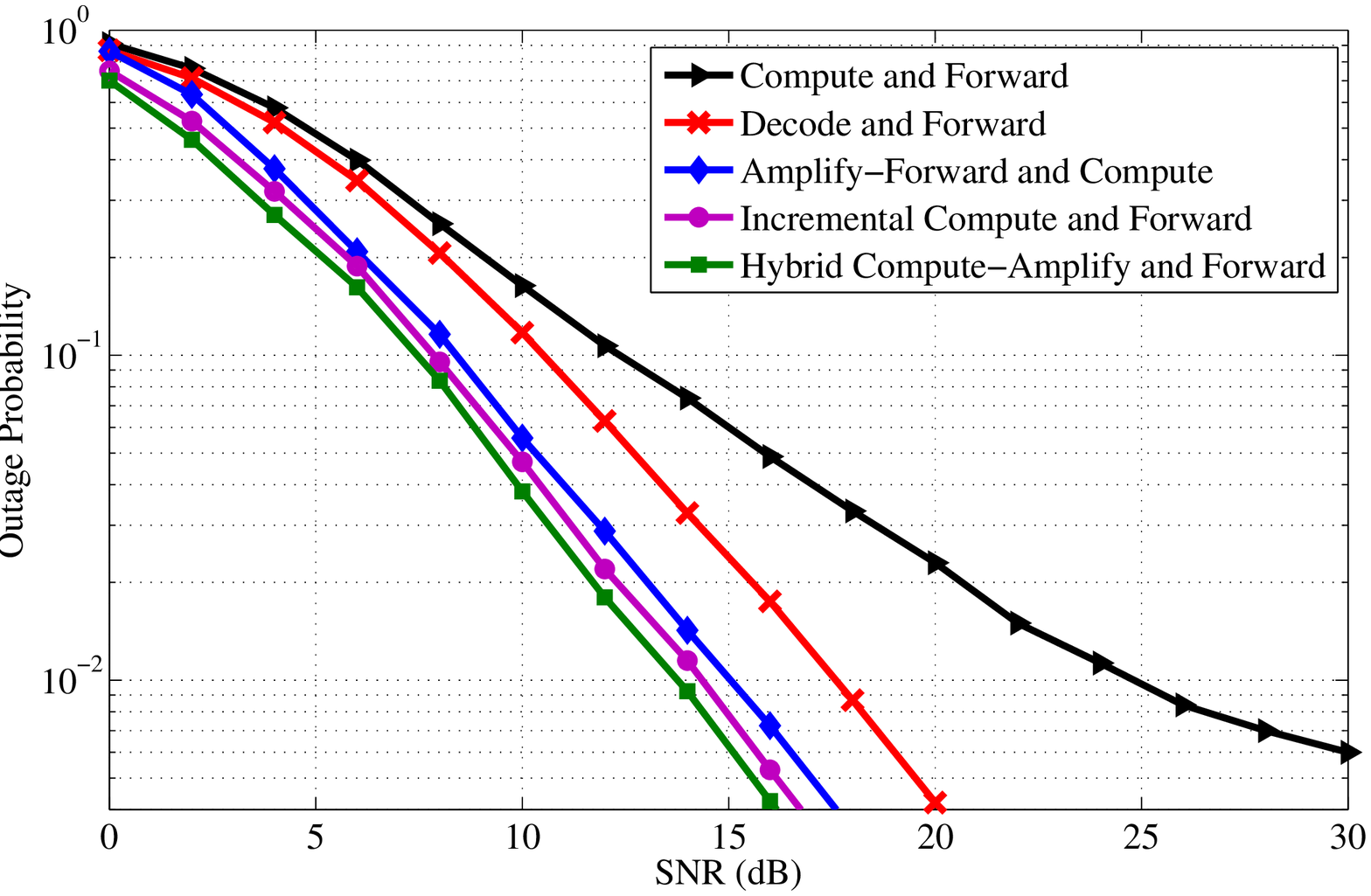}
\caption{Outage probability versus SNR for L=2 and M=3, ${R_t} = 1,\sigma _h^2 = 1,\sigma _f^2 = 10,\sigma _g^2 = 1$.}
\end{figure}


\begin{figure}[t]
\centering
\center
\includegraphics[width =6in]{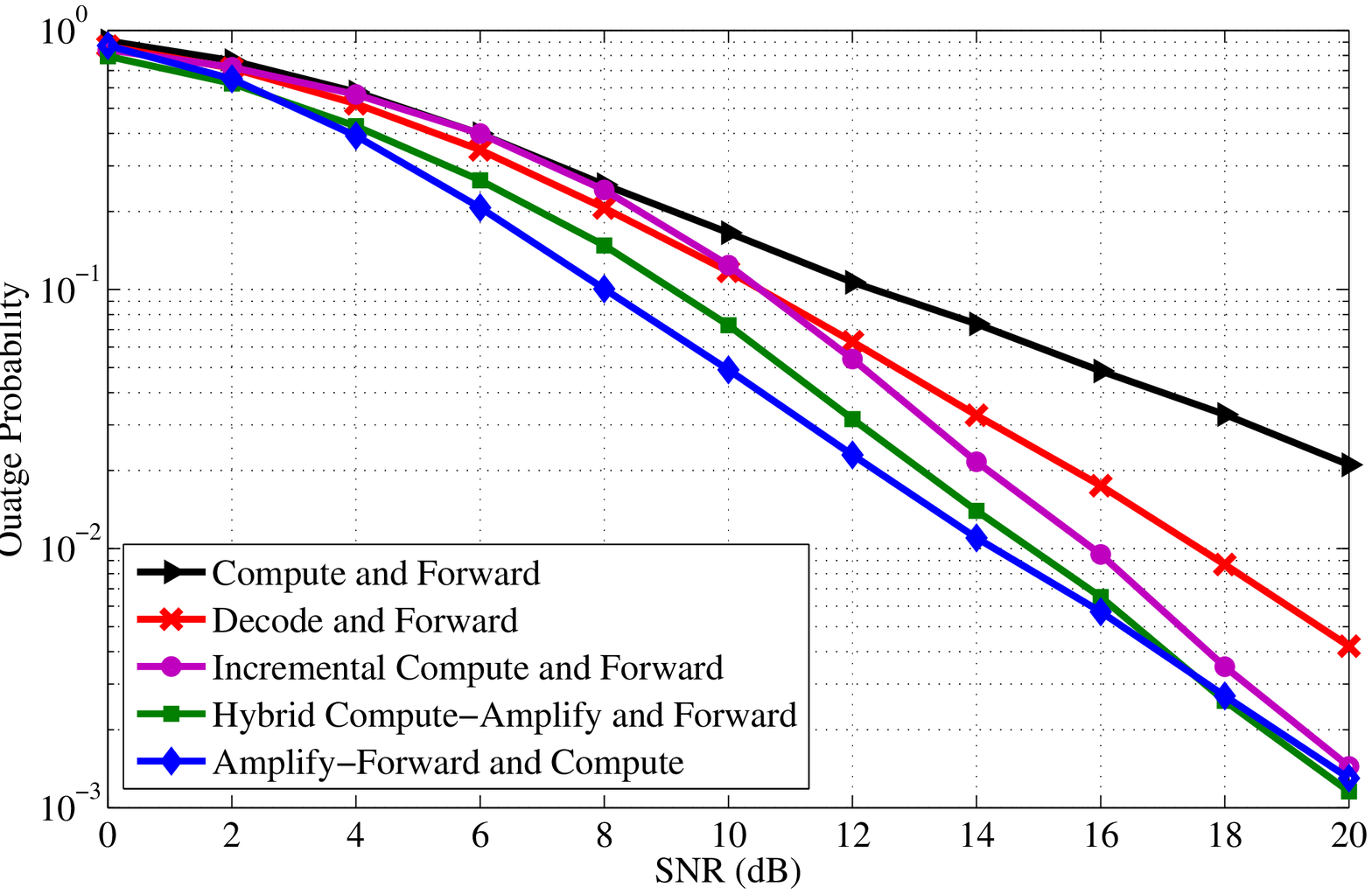}
\caption{Outage probability versus SNR for L=2 and M=3, ${R_t} = 1,\sigma _h^2 = 1,\sigma _f^2 = 10,\sigma _g^2 = 0.1$.}

\end{figure}


\begin{figure}[t]
\centering
\center
\includegraphics[width =6in]{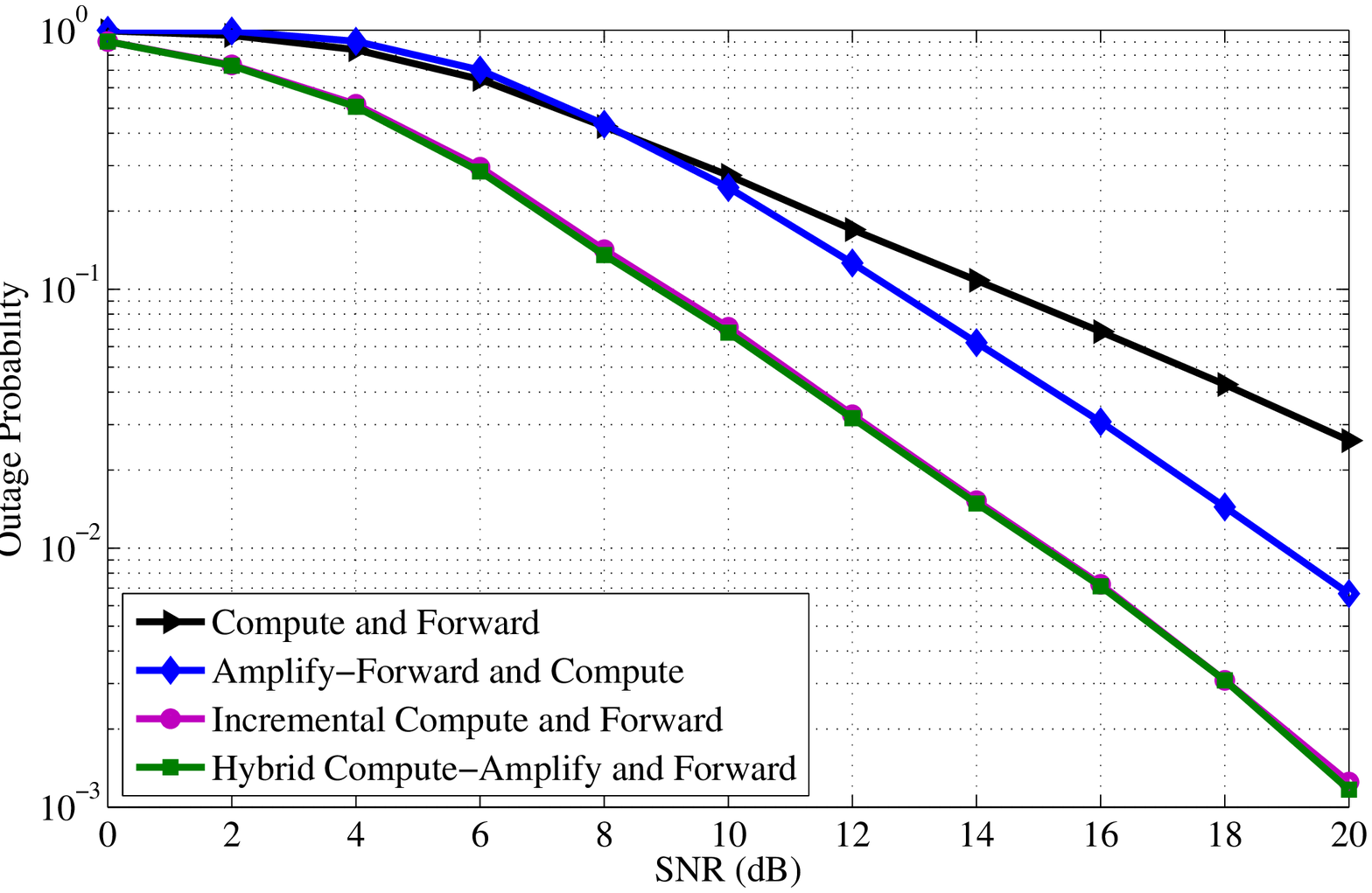}
\caption{Outage probability versus SNR for L=2 and M=3, ${R_t} = 1,\sigma _h^2 = 1,\sigma _f^2 = 1,\sigma _g^2 = 1$.}

\end{figure}
\begin{figure}[t]
\centering
\center
\includegraphics[width =6in]{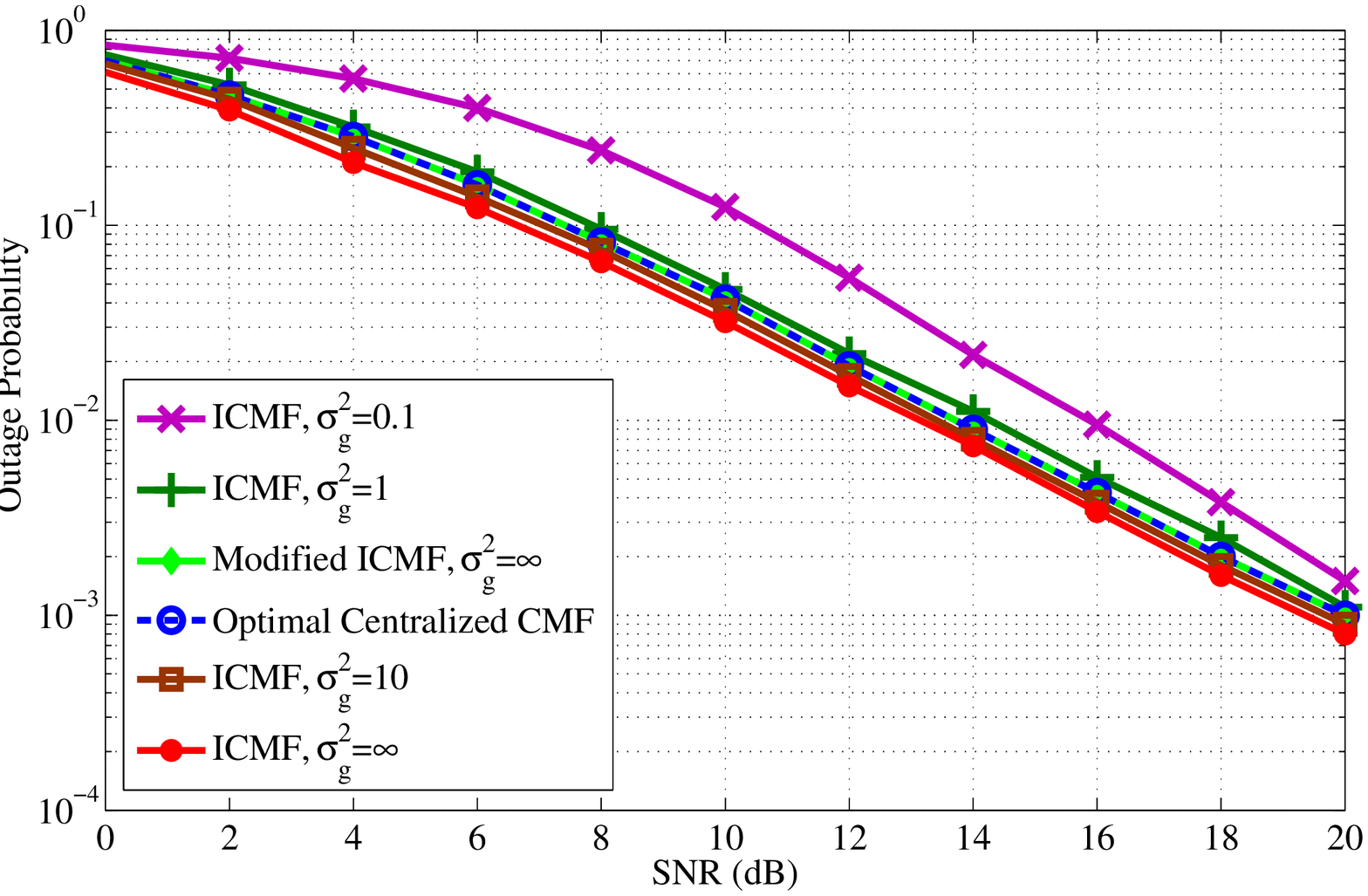}
\caption{Outage probability versus SNR for L=2 and M=3, ${R_t} = 1,\sigma _h^2 = 1,\sigma _f^2 = 10$.}
\end{figure}


\begin{figure}[t]
\centering
\center
\includegraphics[width =6in]{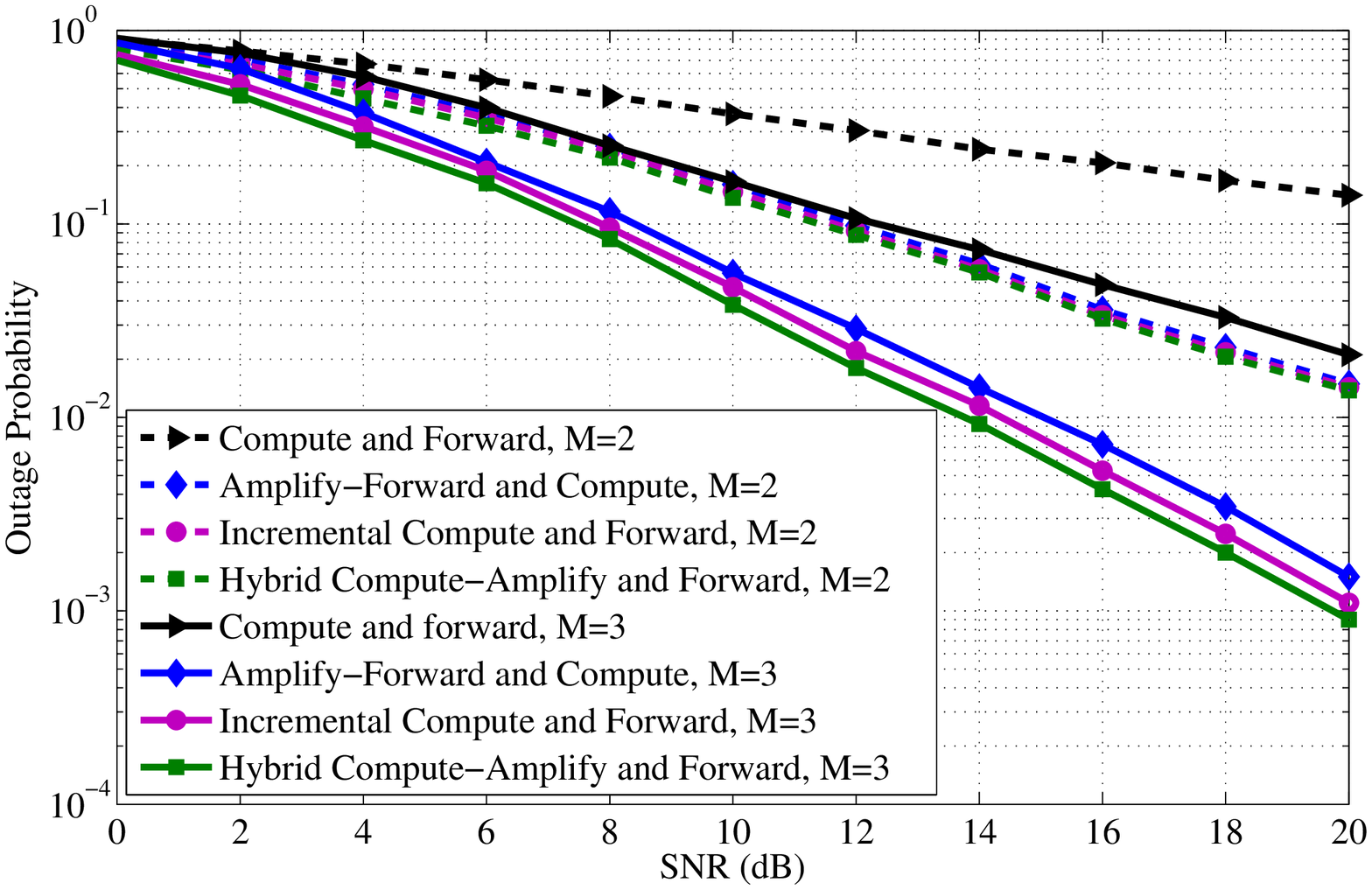}
\caption{Outage probability versus SNR for L=2, and M=2 and 3, ${R_t} = 1,\sigma _h^2 = 1,\sigma _f^2 = 10,\sigma _g^2 = 1$.}

\end{figure}
\begin{figure}[t]
\centering
\center
\includegraphics[width =6in]{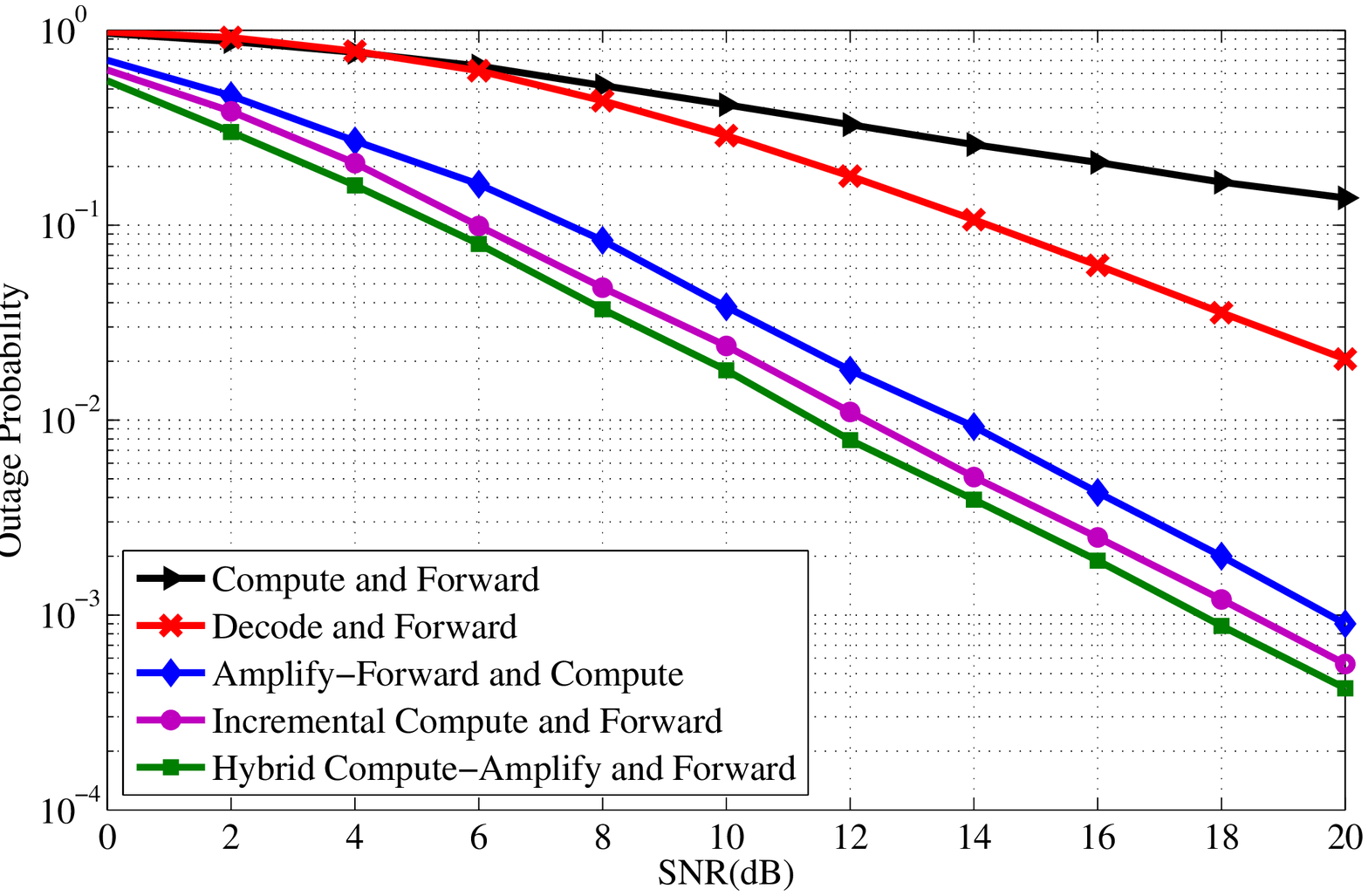}
\caption{Outage probability versus SNR for L=3 and M=3, ${R_t} = 1,\sigma _h^2 = 1,\sigma _f^2 = 10,\sigma _g^2 = 1$.}

\end{figure}

\newpage

\section{Simulation Results}
In this section, we evaluate and compare the performance of our proposed methods through computer simulations. We assume the case in which all of the nodes, i.e., the users and the relays, have equal transmission powers, and $\sigma _{im}^2 = \sigma _h^2 , \forall i,m$, $\sigma _{r,ab}^2 = \sigma _g^2 , \forall a,b$, and $\sigma _m^2 = \sigma _f^2 , \forall m$. However, the same qualitative conclusions as in the presented figures hold for the heterogenous setups as well. Threshold rate is set equal to one ($R_t=1$). 

Figures 4 and 5 show the outage probability, i.e. $\text{Pr}(R_{\text{scheme}}\le R_t))$, versus SNR for three proposed schemes along with those of the conventional CMF and DF relaying schemes for $L=2$, $M=3$, $\sigma _h^2 = 1$, $\sigma _f^2 = 10$, and $\sigma _g^2 = 1$ and $0.1$, respectively. For the DF strategy, the best relay with maximum rate jointly decodes the users' messages utilizing the successive interference cancellation method ([4] and [20]), and then transmits them separately to the destination. From this figures, the ICMF and AFC methods perform significantly better than the CMF and DF methods, especially at high SNRs. For example, for $\sigma _g^2 = 1$ and at outage probability of $10^{-2}$, the proposed ICMF and AFC schemes perform approximately $10 \text{dB}$ and $3 \text{dB}$ better than the CMF and the DF strategies, respectively. Moreover, both the ICMF and AFC methods achieve significantly higher diversity order than CMF in which due to the rank failure problem at the destination, the diversity order is low. As realized from the figures, the HCAF always shows better performance than ICMF; the amount of the improvement decreases with the inter-relay channel qualities, i.e., higher $\sigma _g^2$. For example, at outage probability of $10^{-2}$, for $\sigma _g^2 = 0.1$ and $1$, HCAF outperforms ICMF approximately $1.5 \text{dB}$ and $0.5 \text{dB}$, respectively. Furthermore, ICMF can perform better than the AFC for $\sigma _g^2$ higher than a certain threshold, due to the fact that the ICMF requires each relay to correctly decode the other relays transmissions in order to utilizes the previously transmitted equations. For example at outage $10^{-2}$, while at $\sigma _g^2 = 0.1$, ICMF performs approximately $2 \text{dB}$ worse than the AFC, at $\sigma _g^2 = 1$ it performs $1 \text{dB}$ better. Note although the DF outperforms CMF in these figures, the CMF can have a better performance than DF under very simple scenarios such as two way relay channels, where the probability of rank failure is low [8]. 

In Fig. 6, we consider the case with $L=2$, $M=3$, $\sigma _h^2 = 1$, $\sigma _f^2 = 1$, and $\sigma _g^2 = 1$. By comparison of Figs. 5 and 6, it can be realized that the performance of the proposed schemes is better when the channels from the relays to the destination experience higher SNR, i.e., higher $\sigma _f^2$. As can be observed and expected, the effect of $\sigma _f^2$ on the performance of AFC is more substantial than the other schemes, and the amount of the improvement of ICMF over AFC decreases for high $\sigma _f^2$. For example, at $\sigma _f^2 = 1$ and $10$, and at outage $10^{-2}$, ICMF performs approximately $4 \text{dB}$ and $1 \text{dB}$ better than AFC, respectively.

In Fig. 7, the effect of the inter-relays channels' qualities, i.e. $\sigma _g^2 $, on the performance is considered. In this figure, we have $L=2$, $M=3$, $\sigma _h^2 = 1$, $\sigma _f^2 = 10$. As we can see, ICMF with $\sigma _g^2$ larger than $10$ outperforms the centralized optimal CMF with global knowledge and the modified ICMF (please see Appendix 2 for introduction) performs similar to the optimal CMF. In addition, although the perfromance is degraded by the decrease of $\sigma _g^2 $, the optimal but impractical approach performs only about $2 \text{dB}$ better than the ICMF at very poor inter-relay links, i.e., $\sigma _g^2=0.1 $.

In Fig. 8, the effect of the number of relays on the performance has been studied and compared. The values of the parameters are: $L=2$, $\sigma _h^2 = 1$, $\sigma _f^2 = 10$, $\sigma _g^2 = 1$, and $M=2$ and $3$. By the increase of the number of relays, the performance and also the diversity order are significantly improved. For example at outage $2 \times 10^{-2}$ and the parameter setting of the figure, the proposed schemes with $M=3$ lead to approximately $5.5  \text{dB}$ better than the ones with $M=2$.

In Fig. 9, we consider three users and three relays, i.e. $L=3$ and $M=3$, and we set $\sigma _h^2 = 1$, $\sigma _f^2 = 10$, and $\sigma _g^2 = 1$. At outage probability of $10^{-2}$, it can be observed that the proposed schemes have approximately $8 \text{dB}$ better perfromance than the DF method, and provide significant improvment in comparison with the CMF scheme. As observed from Figs. 4 and 8, the performance gain of the proposed schemes, compared to the state-of-the-art approaches, increases with the number of users.

\section{Conclusion}
In this paper, we considered different relaying strategies for multi-user multi-relay networks, named as ICMF, AFC, and HCAF. In these strategies, new ideas are exploited to overcome the drawbacks of the conventional CMF strategy and to provide efficient and reliable transmission frameworks for multiuser cooperative networks. In ICMF, each relay exploits the previously transmitted equations, in a distributed and cooperative manner, to extract a new independent equation with highest computation rates. In AFC, the relays amplify and forward their received signals and the destination, as a center of computation, recovers all required equations. In HCAF, a combination of  ICMF  and AFC approaches are used in which whenever the highest computation rate of the relays is lower than the target rate, the relays switch from computing nodes to amplifying nodes. 
Numerical results indicate that the outage performance and diversity order of the proposed strategies are considerably better than those of the conventional CMF and DF strategies specially at high number of users or relays. Moreover, numerical results show that ICMF performs better than AFC only when the links among the relays experience high quality. It is notable that the complexity of AFC is much lower than that of the ICMF. HCAF strategy outperforms the ICMF, at the cost of more complicated receiver. Finally, the ICMF and HCAF schemes, independent of the number of relays ($M$), need $L+1$ time slots to transmit the $L$ users' messages, in contrast to AFC and CMF that require $M+1$ time slots. 

\appendices
\section{Proof of Theorem 1}
From (24), the optimum coefficient vectors are obtained by minimizing the following function: 
\begin{eqnarray}
f( {\beta _m^k,\mathbf{c}_m^k}) &=& \frac{1}{{\text{SNR}_T}}{\left| {\beta _m^k} \right|^2} + {\left| {\left| {\beta _m^k\mathbf{g}_m^k + \mathbf{E}_k^*\mathbf{c}_m^k - {\mathbf{a}_l}} \right|} \right|^2}\nonumber\\&=& \frac{1}{{\text{SNR}_T}}\beta _m^k\beta {{_m^k}^*} + {\left( {\beta _m^k\mathbf{g}_m^k + \mathbf{E}_k^*\mathbf{c}_m^k - {\mathbf{a}_l}} \right)^*}\left( {\beta _m^k\mathbf{g}_m^k + \mathbf{E}_k^*\mathbf{c}_m^k - {\mathbf{a}_l}} \right)\nonumber\\ &=& \frac{1}{{\text{SNR}_T}}\beta _m^k\beta {{_m^k}^*} + \beta _m^k\beta {{_m^k}^*}{\left| {\left| {\mathbf{g}_m^k} \right|} \right|^2} + 2\beta {{_m^k}^*}\mathbf{g}{{_m^k}^*}\mathbf{E}_k^*\mathbf{c}_m^k - 2\beta {{_m^k}^*}\mathbf{g}{{_m^k}^*}{\mathbf{a}_l}\nonumber\\ &+& \mathbf{c}{{_m^k}^*}{\mathbf{E}_k}\mathbf{E}_k^*\mathbf{c}_m^k - 2\mathbf{c}{{_m^k}^*}{\mathbf{E}_k}{\mathbf{a}_l} + \mathbf{a}_l^*{\mathbf{a}_l}
\end{eqnarray}
From the definition of $\mathbf{g}_m^k$ in (21), we have,
\begin{eqnarray}
\mathbf{g}{{_m^k}^*}\mathbf{E}_k^*\mathbf{c}_m^k = 0
\end{eqnarray}
Hence, we can write
\begin{eqnarray}
f = \beta _m^k\beta {{_m^k}^*}\left( {\frac{1}{{\text{SNR}_T}}+ {{\left| {\left| {\mathbf{g}_m^k} \right|} \right|}^2}} \right) - 2\beta {{_m^k}^*}\mathbf{g}{{_m^k}^*}{\mathbf{a}_l} + \mathbf{c}{{_m^k}^*}{\mathbf{E}_k}\mathbf{E}_k^*\mathbf{c}_m^k - 2\mathbf{c}{{_m^k}^*}{\mathbf{E}_k}{\mathbf{a}_l} + \mathbf{a}_l^*{\mathbf{a}_l}
\end{eqnarray}
The optimum value for $\beta _m^k$ is obtained by setting the derivative of $f$ with respect to $\beta _m^k$ equal to zero
\begin{eqnarray}
\frac{{\partial f\left( {\beta _m^k,\mathbf{c}_m^k} \right)}}{{\partial \beta _m^k}} = 2\beta _m^k\left( {\frac{1}{{\text{SNR}_T}} + {{\left| {\left| {\mathbf{g}_m^k} \right|} \right|}^2}} \right) - 2\mathbf{g}{{_m^k}^*}{\mathbf{a}_l} = 0
\end{eqnarray}
which leads to:
\begin{eqnarray}
\beta _{m,\text{opt}}^k = \frac{\mathbf{g}{{_m^k}^*}{\mathbf{a}_l}}{\frac{1}{\text{SNR}_T} + || {\mathbf{g}_m^k}||^2}
\end{eqnarray}
In a similar way, to obtain the optimum value for $\mathbf{c}_m^k$, we set:
\begin{eqnarray}
\frac{{\partial f\left( {\beta _m^k,\mathbf{c}_m^k} \right)}}{{\partial \mathbf{c}_m^k}} = 2{\mathbf{E}_k}\mathbf{E}_k^*\mathbf{c}_m^k - 2{\mathbf{E}_k}{\mathbf{a}_l} = 0
\end{eqnarray}
which leads to:
\begin{eqnarray}
\mathbf{c}{{_{m,\text{opt}}^k}^*} = \mathbf{a}_l^*\mathbf{E}_k^*{\left( {{\mathbf{E}_k}\mathbf{E}_k^*} \right)^{ - 1}}
\end{eqnarray}
Thus, the theorem is proved.
\section{Proof of Theorem 2}
Assume that a central coordinator has access to all relays channels information, or equivalently it knows all the equations coefficients recovered all the relays, and selects simultaneously $L$ linearly-independent equations among them. This method, called optimal centralized CMF, is the asymptotic case of the method proposed in [14] when $T_\text{max}$ goes to infinity with considering only $L$ time slots for relays to destination transmission.
Also consider a modified version of our proposed ICMF method, in which the relays do not use the previously selected and transmitted equations for recovering the current equation, in each stage. It is clear that the performance of original ICMF is better than modified ICMF. Note that for the modified ICMF the rate of equation recovered in step $k$ in relay $m$, i.e., $r_m^k$ reduces from (27) to (7).

Here, we prove that the modified ICMF method achieves the same rate as the optimal centralized CMF. Consider that the all equations coefficients known by the central coordinator are sorted in descending order in terms of their computation rates defined in (7), and are denoted by ECVs $\mathbf{c}_k$ and their corresponding computation rates $R(\mathbf{c}_k), \forall k$. The central coordinator searches among all the possible $L$-independent equations combinations, i.e. sets $u$, to choose linearly-independent ECVs $\mathbf{f}_1,...,\mathbf{f}_L$ with maximum overall rate, similar to the processes described in section II, as follows:
\begin{eqnarray}
\left\{ {{\mathbf{f}_1}, \ldots ,{\mathbf{f}_L}} \right\} = arg\mathop {\max }\limits_u {\rm{min}}{\left\{ {R\left( {{\mathbf{c}_{{u_1}}}} \right), \ldots ,R\left( {{\mathbf{c}_{{u_L}}}} \right)}\right\}}
\end{eqnarray}
On the other hand, we have two steps in modified ICMF: first, at each stage, for each relay, the best equation that is independent of the previously selected and transmitted equations is found. Then, using timer setting method, the best equation among the ones recovered by the relays at that stage is selected. This procedure continues until $L$ linearly independent equations ${\mathbf{e}_1,....,\mathbf{e}_L}$ are selected. Hence, when the inter-relays links are strong enough such that the computation rates of all previously transmitted equations at the relays are not decreased, i.e., $r_{mnj}$ in (16) for all $j$ and $m$ is equal or greater than $r_m^k$ in (7), for the last selected equation in modified ICMF, we have
\begin{eqnarray}
{\mathbf{e}_L} = \arg {\max _{m = 1, \ldots ,M}}{\rm{ma}}{{\rm{x}}_{{\mathbf{c}_l^m \independent \left\{{\mathbf{e}_1,...,\mathbf{e}_{L-1}}\right\}}}}{R\left( {\mathbf{c}_l^m} \right)} 
\end{eqnarray}
It is clear that we can rewrite (81) as
\begin{eqnarray}
{\mathbf{e}_L} = \arg {\max _{{\mathbf{c}_k} \independent \left\{ {{\mathbf{e}_1}, \ldots ,{\mathbf{e}_{L - 1}}} \right\}}}R\left( {{\mathbf{c}_k}} \right)
\end{eqnarray}
Now suppose that the optimal approach has a rate higher than the modified ICMF, i.e., we have 

$\text{min}\left\{ {R\left( {{\mathbf{f}_{1}}} \right), \ldots ,R\left( {{\mathbf{f}_{{L}}}} \right)} \right\} > R(\mathbf{e}_L) $. From (82), $\mathbf{e}_L$ is the ECV with the highest rate among all the ECVs that are linearly independent of ${\mathbf{e}_1,...,\mathbf{e}_{L-1}}$. This implies that any ECV with a rate higher than $R(\mathbf{e}_L)$ is linearly dependent to the set ${\mathbf{e}_1,...,\mathbf{e}_{L-1}}$. As a result, we have
\begin{eqnarray}
\left\{ {{\mathbf{f}_1}, \ldots ,{\mathbf{f}_{L}}} \right\} \in \text{span}\left\{ {{\mathbf{e}_1}, \ldots ,{\mathbf{e}_{L - 1}}} \right\}
\end{eqnarray}
Which indicates that the dimension of the space spanned by vectors $\mathbf{f}_1,...,\mathbf{f}_L$ is at most $L-1$, But this contradicts the assumption of linear-independency of these equations. Hence, the modified ICMF achieves the same rate as the optimal centralized CMF, without requiring to collect all equations recovered by all relays in a coordinator center. Moreover, since original ICMF, in each stage, takes advantage of the previously recovered equations in decoding of current equation, it can achieve a rate higher than the modified ICMF and optimal centralized CMF.

\section{Proof of Theorem 3}
From (60), the optimum values are obtained by minimizing the following function:
\begin{eqnarray}
f\left( {{\mathbf{b}_l},{\mathbf{c}_l}} \right) &=& \frac{1}{{\text{SNR}_T}}{\left| {\left| {{\mathbf{b}_l}} \right|} \right|^2} + \frac{1}{{\text{SNR}_T}}{\left| {\left| {{\mathbf{F}^{\text{AF}}}^*{\mathbf{b}_l}} \right|} \right|^2} + {\left| {\left| {{\mathbf{G}^*}{\mathbf{b}_l} + {\mathbf{D}^*}{\mathbf{c}_l} - {\mathbf{a}_l}} \right|} \right|^2}\nonumber\\ &=& \frac{1}{{\text{SNR}_T}}\mathbf{b}_l^*{\mathbf{b}_l} + \frac{1}{{\text{SNR}_T}}\mathbf{b}_l^*{\mathbf{F}^{\text{AF}}}{\mathbf{F}^{\text{AF}}}^*{\mathbf{b}_l} + {\left( {{\mathbf{G}^*}{\mathbf{b}_l} + {\mathbf{D}^*}{\mathbf{c}_l} - {\mathbf{a}_l}} \right)^*}\left( {{\mathbf{G}^*}{\mathbf{b}_l} + {\mathbf{D}^*}{\mathbf{c}_l} - {\mathbf{a}_l}} \right)\nonumber\\ &=& \frac{1}{{\text{SNR}_T}}\mathbf{b}_l^*\left( {\mathbf{I} + {\mathbf{F}^{\text{AF}}}{\mathbf{F}^{\text{AF}}}^*} \right){\mathbf{b}_l} + \mathbf{b}_l^*\mathbf{G}{\mathbf{G}^*}{\mathbf{b}_l} + 2\mathbf{b}_l^*\mathbf{G}{\mathbf{D}^*}{\mathbf{c}_l} - 2\mathbf{b}_l^*\mathbf{G}{\mathbf{a}_l}\nonumber\\ &+& \mathbf{c}_l^*\mathbf{D}{\mathbf{D}^*}{\mathbf{c}_l} - 2\mathbf{c}_l^*\mathbf{D}{\mathbf{a}_l} + \mathbf{a}_l^*{\mathbf{a}_l}
\end{eqnarray}
 From the definition of $\mathbf{G}$ in (56), we easily obtain:
\begin{eqnarray}
\mathbf{b}_l^*\mathbf{G}{\mathbf{D}^*}{\mathbf{c}_l} = 0
\end{eqnarray}
Hence, we have
\begin{eqnarray}
f = \mathbf{b}_l^*\left( {\frac{1}{{\text{SNR}_T}}\left( {\mathbf{I} + {\mathbf{F}^{\text{AF}}}{\mathbf{F}^{\text{AF}}}^*} \right) + \mathbf{G}{\mathbf{G}^*}} \right){\mathbf{b}_l} - 2\mathbf{b}_l^*\mathbf{G}{\mathbf{a}_l} + \mathbf{c}_l^*\mathbf{D}{\mathbf{D}^*}{\mathbf{c}_l} - 2\mathbf{c}_l^*\mathbf{D}{\mathbf{a}_l} + \mathbf{a}_l^*{\mathbf{a}_l}
\end{eqnarray}
The optimum value of $\mathbf{b}_l$ is the solution of
\begin{eqnarray} 
\frac{{\partial f\left( {{\mathbf{b}_l},{\mathbf{c}_l}} \right)}}{{\partial {\mathbf{b}_l}}} = 2\left( {\frac{1}{{\text{SNR}_T}}\left( {\mathbf{I} + {\mathbf{F}^{\text{AF}}}{\mathbf{F}^{\text{AF}}}^*} \right) + \mathbf{G}{\mathbf{G}^*}} \right){\mathbf{b}_l} - 2\mathbf{G}{\mathbf{a}_l} = 0
\end{eqnarray}
Hence,
\begin{eqnarray}
\mathbf{b}_{\text{opt},l}^* = \mathbf{a}_l^*{\mathbf{G}^*}{\left( {\frac{1}{{\text{SNR}_T}}\left( {\mathbf{I} + {\mathbf{F}^{\text{AF}}}{\mathbf{F}^{\text{AF}}}^*} \right) + \mathbf{G}{\mathbf{G}^*}} \right)^{ - 1}}
\end{eqnarray}
In a similar way, the optimum value of $\mathbf{c}_l$ is found from the solution of  
\begin{eqnarray}
\frac{{\partial f\left( {{\mathbf{b}_l},{\mathbf{c}_l}} \right)}}{{\partial {\mathbf{c}_l}}} = 2\mathbf{D}{\mathbf{D}^*}{\mathbf{c}_l} - 2\mathbf{D}{\mathbf{a}_l} = 0
\end{eqnarray}
which leads to
\begin{eqnarray}
\mathbf{c}_l^* = \mathbf{a}_l^*{\mathbf{D}^*}{\left( {\mathbf{D}{\mathbf{D}^*}} \right)^{ - 1}}
\end{eqnarray}
Thus, the theorem is proved.

\end{document}